\documentclass[9pt,shortpaper,twoside,web]{ieeecolor}
\usepackage{generic}
\usepackage{cite}
\usepackage{amsmath,amssymb,amsfonts}
\usepackage{graphicx}
\usepackage{textcomp}
\usepackage{subfigure}
\usepackage{subeqnarray}
\usepackage{mathrsfs}
\usepackage{multirow}
\usepackage{diagbox}
\usepackage{siunitx}

\usepackage{algorithm}
\usepackage{algpseudocode}
\usepackage{amsmath}
\usepackage{epsfig}
\usepackage{color,xcolor}

\newtheorem{theorem}{Theorem}

\newtheorem{remark}{Remark}
\newtheorem{assumption}{Assumptions}

\def\BibTeX{{\rm B\kern-.05em{\sc i\kern-.025em b}\kern-.08em
    T\kern-.1667em\lower.7ex\hbox{E}\kern-.125emX}}
\begin{document}
\title{Robust Kalman filters with unknown covariance of multiplicative noise}
\author{Xingkai Yu and Ziyang Meng, ~\IEEEmembership{Senior Member,~IEEE}
\thanks{This work has been supported in part by the National Natural Science Foundation of China under Grants U19B2029 and 62103222. (Corresponding author: Ziyang Meng.)}
\thanks{X. Yu and Z. Meng are with the  Department of Precision Instrument, Tsinghua University, Beijing, 100084, China (e-mail: xkyu2007@mail.tsinghua.edu.cn; ziyangmeng@mail.tsinghua.edu.cn). }
}

\maketitle

\begin{abstract}
In this paper, the joint estimation of state and noise covariance for linear systems with unknown covariance of multiplicative noise is considered. The measurement likelihood is modelled as a mixture of two Gaussian distributions and a Student's $\emph{t}$ distribution, respectively. The unknown covariance of multiplicative noise is modelled as an inverse Gamma/Wishart distribution and the initial condition is formulated as the nominal covariance. By using robust design and choosing hierarchical priors, two variational Bayesian (VB) based robust Kalman filters are proposed. The stability and convergence of the proposed filters and the covariance parameters are analyzed. The lower and upper bounds are also provided to guarantee the performance of the proposed filters. A target tracking simulation is provided to validate the effectiveness of the proposed filters.

\end{abstract}

\begin{IEEEkeywords}
Kalman filter, multiplicative noise, unknown covariance, variational Bayesian (VB).
\end{IEEEkeywords}

\section{Introduction}

The Kalman filter (KF) is widely used in state estimation problems and is shown to be optimal in the sense of minimum variance for linear systems with Gaussian white additive measurement and process noises \cite{1100100}. On the one hand,
the usefulness of the Kalman filter is limited by  the \textit{prior} information of noise statistics. The inaccurate usage of prior information may result in bad performance and even divergence. On the other hand, in many practical applications, the prior information may be unknown and then the robust adaptive Kalman filter is introduced to solve such a problem, which can be classified into four branches including maximum likelihood, Bayesian, covariance matching, and correlation approaches. 
The typical implements of adaptive KF include \textit{innovation-based adaptive filter} \cite{karasalo2011optimization}, 
\textit{interacting multiple model filter},  and the \textit{Sage-Husa adaptive filter}  \cite{6179988}. In particular, 
\textit{innovation-based adaptive filter} takes advantage of the innovation sequence and a maximum likelihood criterion to estimate the noise covariance. \textit{Interacting multiple model filter} is a Bayesian approach, while other methods can be seen as its approximations.  \textit{Sage-Husa adaptive filter}  leverages the covariance matching approach and the maximum-a-posterior criterion to estimate the noise statistics recursively. This class of adaptive KF may suffer from issues of convergence, practical applications, and computation burden \cite{sarkka2009recursive}. Therefore, many improved robust adaptive filters were subsequently introduced. In particular, generalized robust noise-identification filters are proposed by using Huber's theory and the expectation maximization algorithm \cite{schon2011system}. A maximum correntropy Kalman filter is proposed to address the heavy-tailed noise problem \cite{chen2017maximum}. However, the aforementioned filters need the known nominal covariance as prior information. Then, Student's $\emph{t}$ filters \cite{huang2016robust}, robust Student's  $\emph{t}$ filters \cite{huang2017novel1} and  Kullback-Leibler divergence based Student's $\emph{t}$ filters \cite{huang2019novel} are proposed to address this problem.  

More recently, a variational Bayesian (VB) based adaptive filter has received much attention since it can be used to perform approximate posterior inference and to estimate uncertain hidden parameters or state variables \cite{4585346} in broad applications including machine learning, visual tracking, signal processing, etc \cite{ji2006variational}. In fact, using appropriate conjugate prior distributions, the existing VB approaches are provided to approximate the additive measurement/process noise covariance matrices \cite{sarkka2009recursive,agamennoni2012approximate}. Due to the strong adaptability of VB, a large number of research results are subsequently obtained. These existing VB based filters usually model heavy-tailed non-Gaussian/Gaussian additive process and measurement noises with the Student's $\emph{t}$-distribution (Std) and model additive noise covariance with the Wishart distribution, the inverse Wishart distribution, or the inverse Gamma distribution \cite{huang2017novel,zhu2021novel}.

To the best of the authors' knowledge, all aforementioned robust adaptive Kalman filters (KFs) are provided to address  state and noise covariance  joint estimation problems for additive noises but are not applicable to the case of  multiplicative noise. In fact, \textit{multiplicative} noise is ubiquitous in practical applications including uncertain measurements and fading or reflection of the transmitted signal over channels \cite{wang2020novel,gao2019robust,yu2019estimation}. Since the product of two Gaussian distributions (multiplicative noise and state) is shown to be a compressed or amplified Gaussian distribution (for the same variable) and non-Gaussian distribution (for different variables) \cite{liu2020quadratic,springer1966distribution},
this difficult issue has attracted widespread attention in different fields \cite{yu2019estimation}. Although 
there are many robust adaptive KFs that can estimate the unknown covariance of additive noise, ones still have no suitable solution to estimate the unknown covariance of multiplicative noise. 

The main contributions of this paper are summarized as follows.
 1) In this paper, the Student's $\emph{t}$ distribution and a mixture of two Gaussian (MtG) distributions, in which one is used to depict the additive noise and the other is to depict the multiplicative noise, are proposed to model measurement likelihood and to perform state estimation with unknown covariance of multiplicative noise. The distribution of multiplicative covariance is chosen as an inverse Gamma/Wishart distribution. The unknown parameters (states and covariance parameters) are jointly estimated in the VB framework. 2) Compared with the existing VB filters for additive noises \cite{huang2016robust}, we propose an improved Student's $\emph{t}$-distribution based Kalman filter \cite{huang2017novel1} that is applicable to multiplicative noise. To solve the technical issue raised by the multiplicative noise, the variational Bayesian inference is employed. Moreover, the proposed filter is capable of estimating the covariances of multiplicative and additive measurement noises as a whole, while the filters proposed in \cite{yu2019estimation} often operate in a separate way. Towards this end,  two robust Kalman filters are presented to jointly estimate the covariance of multiplicative noise and states based on the VB inference. 3) The stability and convergence of the proposed filters, the covariance parameters, and the VB inference are analyzed. The lower and upper bounds are also provided to guarantee the performance of the proposed filters.

The remainder of this paper is organized as follows. Sections II and III state the considered problem and the proposed filters, respectively. Then, Section IV evaluates the performance and 
Section V gives a numerical simulation example. Finally, Section VI concludes the paper.

\textbf{Notation.}  $\mathrm{~N}(\mathbf{z}_{k} ;  \mathbf{x}_{k}, \tilde{\mathbf{R}}_{k})$ denotes that stochastic vector $\mathbf{z}_{k}$ obeys a Gaussian distribution with a mean vector $\mathbf{x}_{k}$ and a covariance matrix $\tilde{\mathbf{R}}_{k}$. $\mathrm{G}(\lambda_{k} ; \frac{1}{2}v, \frac{1}{2}v)$ denotes that variable $\lambda_{k}$ obeys a Gamma distribution with a dof-parameter variable $v$. $\operatorname{IG}(\sigma_{k} ; {\alpha}_{k}, {\beta}_{k})$ denotes that variable $\sigma_{k}$ obeys an inverse Gamma distribution with a shape parameter ${\alpha}_{k}$  and a scale parameter ${\beta}_{k}$. $\operatorname{IW}({{\bf{R}}}_{k} ; {u}_{k}, {\mathbf{U}}_{k})$ denotes that stochastic matrix ${{\bf{R}}}_{k}$ obeys an  inverse Wishart distribution with a dof-parameter variable ${u}_{k}$ and a scale matrix ${\mathbf{U}}_{k}$. $\operatorname{St}(\mathbf{z}_{k} ;  \mathbf{x}_{k}, {\bf{R}}_{k}, v)$ denotes that stochastic variable $\mathbf{z}_{k}$ obeys a Student's  $\emph{t}$ distribution with a location parameter $\mathbf{x}_{k}$, a scale parameter ${\bf{R}}_{k}$ and  a dof-parameter variable $v$. $\hat{\mathbf{x}}_{k | k-1}$, $\hat{\mathbf{x}}_{k | k}$, ${\bf{{{P}}}}_{k|k-1}$, and ${\bf{{{P}}}}_{k|k}$ denote the estimated values of state $\mathbf{x}_{k}$ in the time-update step, the values in the measurement-update step, and their corresponding estimation error covariances, respectively. $\mathrm{E}[\cdot]$, $\log[\cdot]$, and $\text{tr}(\cdot)$ denote the expectation operation, the logarithmic function, and the trace operation, respectively. $p(\bf{x})$ and $q(\bf{x})$ respectively denote the distribution of $\bf{x}$ and the approximated distribution of $p(\bf{x})$. Finally, $\|\cdot\|$ and  $|\cdot|$ denote, respectively,  the 2-Euclid-norm  and the determinant of a matrix or a vector.

\section{Problem Formulation}\label{1}

 Consider a discrete stochastic system contaminated by multiplicative noise described by the following state-space model
\begin{align}
{\bf{x}}_{k} &= {\bf{F}}{\bf{x}}_{k-1} + {\bf{w}}_{k-1}, \label{eq1}\\
{\bf{z}}_{k} &= m_{k}{\bf{H}}{\bf{x}}_{k} + {\bf{v}}_{k}, \label{eq2}
\end{align}
where 
${\bf{x}}_{k} \in {{\mathbb{R}}^n}$ denotes the state vector, ${\bf{z}}_{k} \in {{\mathbb{R}}^m}$ denotes the measurement vector, ${\bf{F}}\in {{\mathbb{R}}^{n\times n}}$ and  ${\bf{H}}\in {{\mathbb{R}}^{m\times n}}$ denote the given state transition matrix and the observation matrix, respectively, and ${\bf{w}}_{k}\in {{\mathbb{R}}^n}$ and ${\bf{v}}_{k}\in {{\mathbb{R}}^m}$ are mutually uncorrelated  Gaussian white additive noises with zero means and covariance matrices ${\bf{Q}}_{k}$ and ${\bf{R}}_{k}$, respectively. In addition, $m_{k}\in {{\mathbb{R}}}$ is an unknown Gaussian multiplicative noise with mean $\bar{m}_{k}$ and covariance ${\bf{\sigma}}_{k}$. 

\textbf{ Measurement likelihood model}. Here, $m_{k}$ is introduced and leads to the non-Gaussian change of measurement likelihood. To characterize this non-Gaussian property, in this paper, the measurement likelihood is modelled as an  MtG distribution model or a Student's $\emph{t}$ distribution (Std) model respectively given by
\begin{align}
p(\mathbf{z}_{k} | \mathbf{x}_{k})&=\mathrm{~N}(\mathbf{z}_{k} ; 0, \tilde{\mathbf{R}}_{k})+ \mathrm{N}\left(\mathbf{z}_{k} ; \bar{m}_{k} \mathbf{H} \mathbf{x}_{k}, \check{\bf{R}}_{k}\right),\label{eq3}\\
 p(\mathbf{z}_{k} | \mathbf{x}_{k})&=\operatorname{St}(\mathbf{z}_{k} ; \bar{m}_{k}\mathbf{H} \mathbf{x}_{k}, \bar {\bf{R}}_{k}, v)\notag \\
&=\mathrm{N}(\mathbf{z}_{k} ; \bar{m}_{k}{\bf{H}}{\bf{x}}_{k}, \bar {\bf{R}}_{k} / \lambda_{k}) \mathrm{G}(\lambda_{k} ; \frac{v}{2}, \frac{v}{2}),\label{eq4}
\end{align}  
where $v$ is the dof-parameter variable, $\lambda_{k}$ denotes scale-parameter variable and $\lambda_{k}\ll1$. Here, $\bar{\bf{R}}_{k}$, $\tilde{\mathbf{R}}_{k}$, and 
$\check{\bf{R}}_{k}$ are respectively the covariances of the mixing total measurement noise, multiplicative measurement noise (in MtG), and the additive measurement noise as 
 \begin{align}
 \bar{\bf{R}}_{k}&=\tilde{\bf{R}}_{k}+\check{\bf{R}}_{k},\tilde{\bf{R}}_{k}= {\bf{\sigma}}_{k}{\bf{H}}{\bf{{{S}}}}_{k}{\bf{H}}^{\top},\label{eq5}\\
  \check{\bf{R}}&={\bar{m}_{k}^2}{\bf{H}}{\bf{{{P}}}}_{k|k-1}{\bf{H}}^{\top}+{\bf{R}}_{k},  {\bf{{{S}}}}_{k}=\text{E}[\mathbf{x}_{k}\mathbf{x}_{k}^{\top}].\label{eq6}
 \end{align}
In order to better regulate the likelihood, the prior on $\lambda_{k}$ is chosen as $p(\lambda_{k})=\mathrm{G}(\lambda_{k} ; \frac{1}{2}v, \frac{1}{2}v)$. In addition, Gamma prior on $v$ is introduced as a Gamma distribution $p(v)= \mathrm{G}(v ; c_0, d_0)$. 
 
\begin{remark}
\emph{The reasons that we choose likelihood  $p({\bf{z}}_{k}|{\bf{x}}_{k})$ as a Student's t distribution or an MtG distribution are as follows. In the first place, it is well known that the product of Gaussian probability density functions (PDFs) for different variables is generally the Meijer G function  \cite{bromiley2003products,springer1966distribution}. This means that the PDF of $m_{k} \times {\bf{x}}_{k}$ in (2) is Meijer G, and therefore the likelihood  $p({\bf{z}}_{k}|{\bf{x}}_{k})$ is non-Gaussian. 
Meijer G functions are usually not analytically tractable. 
Therefore, we choose an approximation method.  On the other hand, we also know that an infinite mixture of Gaussian PDFs can approach arbitrary distribution \cite{bishop2006pattern}.  Student's ${t}$ is an infinite mixture of Gaussian PDFs and therefore can be used to represent a Meijer G-function. In the second place, according to the definition of covariance of measurement equation, MtG is a mixture of two Gaussian distributions, where    $\mathrm{N}\left(\mathbf{z}_{k} ; \bar{m}_{k} \mathbf{H} \mathbf{x}_{k}, \check{\bf{R}}_{k}\right)$ is used to model the disturbance of multiplicative noise and the other $\mathrm{~N}(\mathbf{z}_{k} ; 0, {\tilde{\mathbf{R}}_{k})}$ is to model the influence of additive measurement noise.}
\end{remark}

\textbf{ Covariance model}. For the two models (\ref{eq3})-(\ref{eq4}), we characterize the multiplicative noise covariance as a random variable, which respectively obeys an \emph{inverse Gamma} distribution and an \emph{inverse Wishart} distribution. For model (\ref{eq3}), we directly model the the multiplicative noise covariance $\sigma_{k}$. While, for model (\ref{eq4}), we model the unknown covariance of measurement noises ($m_{k}$ and ${\bf{v}}_{k}$) as a whole ${\bar{\bf{R}}}_{k}$. 

First, we define  $p(\sigma_{k} | \mathbf{z}_{1: k-1}) =\operatorname{IG}(\sigma_{k} ; \hat{\alpha}_{k | k-1}, \hat{\beta}_{k | k-1})$ (If $m_{k}\in {{\mathbb{R}}^{m\times m}}$, then choose \emph{inverse Wishart} distribution (IW), since IW is the general matrix case of IG \cite{sarkka2009recursive}). Let the posterior of $\sigma_{k-1} $ be $p(\sigma_{k-1} | \mathbf{z}_{1: k-1}) =\operatorname{IG}(\sigma_{k-1} ; \hat{\alpha}_{k-1 | k-1}, \hat{\beta }_{k-1 | k-1})$. This paper uses the similar heuristics as given in \cite{sarkka2009recursive}, i.e., 
\begin{align}
\hat{\alpha}_{k | k-1}&=\rho\hat{\alpha}_{k-1 | k-1}, \label{eq7}\\
\hat{{\beta}}_{k | k-1}&=\rho \hat{{\beta}}_{k-1 | k-1},\label{eq8}
\end{align} 
where $\rho \in(0,1]$. The initial value is chosen as $\frac{\hat{\mathbf{\beta}}_{0 | 0}}{\hat{\alpha}_{0 | 0}}={\sigma}_{0}$, where ${\sigma}_{0}$ is an empirical constant.

Then, similarly, define  $p({\bar{\bf{R}}}_{k} | \mathbf{z}_{1: k-1}) =\operatorname{IW}({\bar{\bf{R}}}_{k} ; \hat{u}_{k | k-1}, \hat{\mathbf{U}}_{k | k-1})$ and choose the similar heuristics  in \cite{huang2017novel} as 
\begin{align}
\hat{u}_{k \mid k-1}&=\rho(\hat{u}_{k-1 \mid k-1}-m-1)+m+1, \label{eq9}\\
\hat{\mathbf{U}}_{k \mid k-1}&=\rho \hat{\mathbf{U}}_{k-1 \mid k-1}, \label{eq10}
\end{align}
with the initial value being 
$
\frac{\hat{\mathbf{U}}_{0 | 0}}{\hat{u}_{0|0}-m-1}=\bar{\mathbf{R}}_{0},
$
where $\bar{\mathbf{R}}_{0}$ is also a chosen initial value.

Before presenting the main results, we formally provide the following assumption. 
\begin{assumption}
\emph{Both the additive process noise ${\bf{w}}_{k}$ and the additive measurement noise ${\bf{v}}_{k}$ are zero-mean Gaussian white noises.  The multiplicative measurement noise $m_{k}$ is a non-zero mean Gaussian noise. The initial state vector ${\bf{x}}_{0}$ is assumed to be a Gaussian distribution with mean $\hat{\mathbf{x}}_{0 | 0}$ and covariance matrix ${\bf{{{P}}}}_{0|0}$. The initial state ${\bf{x}}_{0}$  and the noise sequences ($\{{\bf{w}}_{k}\}$, $\{{\bf{v}}_{k}\}$, and $\{m_{k}\}$) are mutually independent. }
\end{assumption}

\textbf{One-step prediction}. The derivation of the desired filter includes one-step prediction (time update) step and measurement update step. Since the existence of multiplicative noise does not affect the state equation (\ref{eq1}), according to the standard Kalman filter, 
the one-step prediction of state is modelled as a Gaussian distribution 
\begin{equation}
p({\bf{x}}_{k}|{\bf{z}}_{1:k-1}) =\text{N}({\bf{x}}_{k}; {\bf{{\hat{x}}}}_{k|k-1}, {\bf{{{P}}}}_{k|k-1} ), 
\end{equation} 
where  one-step prediction  ${\bf{{\hat{x}}}}_{k|k-1}$ and its corresponding error covariance matrix  ${\bf{{{P}}}}_{k|k-1}$ are given by 
\begin{align}
{\bf{{\hat{x}}}}_{k|k-1} &= {\bf{F}}{\bf{{\hat{x}}}}_{k-1|k-1}, \label{eq12} \\
{\bf{{{P}}}}_{k|k-1} &={\bf{F}}{\bf{{{P}}}}_{k-1|k-1}{\bf{F}}^{\top} + {\bf{Q}}_{k-1}. \label{eq13}
\end{align}
The system model, the proposed unknown covariance distribution model, and the one-step prediction are formulated as in (\ref{eq1})-(\ref{eq13}). We next present the measurement update and then form the whole filter in the following section.
\begin{remark}
\emph{The considered model (\ref{eq2}) can be used to model many practical applications in the fields of communication, signal processing, petroleum seismic exploration, target tracking, and so on \cite{liu2015optimal,yang2013optimal}. Different from additive measurement noise, the second- and high-order statistics of multiplicative measurement noise are usually unknown and  difficult to estimate because they trigger additional difficulties and lead to great fluctuations for state signals. }
\end{remark}

\section{The presented filters}\label{section Problem Formulation}

In this section, we first introduce the VB inference, which will be used together with the fixed-point iteration method in the derivation of the desired filter. Then, we present two filters to address the problem of unknown multiplicative noise covariance, one based on the assumption that the likelihood function is the MtG model and the other on the Std model.

\subsection{VB inference}\label{section Problem Formulation1}

Based on Bayes' rule,  we have
$p(\mathbf{\varPhi}_{k} | \mathbf{z}_{1: k}, \varPsi_{k})=\frac{p\left(\mathbf{z}_{1: k} | \varPhi_{k}\right) p\left(\varPhi_{k} | \varPsi_{k}\right)}{p\left(\mathbf{z}_{1 \cdot k} | \varPsi_{k}\right)}$, where $\varPhi_{k}=\{{\bf{x}}_{k},{\bf{\sigma}}_{k}, {\bar{\bf{R}}}_{k}\}$ and $\varPsi_{k}=\{\hat{\alpha}_{k | k-1}, \hat{\beta}_{k | k-1},  \hat{{u}}_{k | k-1}, \hat{\mathbf{U}}_{k | k-1}\}$. Since we cannot obtain the analytical solution of $p(\mathbf{\varPhi}_{k} | \mathbf{z}_{1: k}, \varPsi_{k})$, $p(\mathbf{\varPhi}_{k} | \mathbf{z}_{1: k}, \varPsi_{k})$ is approximated using a distribution $q(\mathbf{\varPhi}_{k})$. Based on variational inference \cite{bishop2006pattern}, it yields 
\begin{equation}
\log q(\theta) =\mathrm{E}_{{\varPhi_{k}}^{(-\theta)}}\left[\log p\left(\mathbf{z}_{1: k} | \varPhi_{k}\right) p\left(\varPhi_{k} | \varPsi_{k}\right)\right], \label{eq14}
\end{equation}
where $\theta$ is a member of $\varPhi_{k}$ and  ${-\theta}$ represents its complementary set. Then, by taking $\varPhi_{k}=\{{\bf{x}}_{k},{\bf{\sigma}}_{k}\}$ as an example, we have 
\begin{scriptsize}
\begin{align}\label{eq15}
p\left(\mathbf{z}_{1: k} | \varPhi_{k}\right) p\left(\varPhi_{k} | \varPsi_{k}\right)
=p(\mathbf{z}_{1: k-1}) p({{\sigma}}_{k} |{\bf{z}}_{1:k-1}) p({\bf{x}}_{k}|{\bf{z}}_{1:k-1}) p(\mathbf{z}_{k} | \mathbf{x}_{k}, {{\sigma}}_{k}).
\end{align}
\end{scriptsize}
Since (\ref{eq15}) cannot be solved directly, we apply the fixed-point iteration approach to solve it, i.e., $q(\theta)$ is updated as $q^{(i+1)}(\theta)$ at the ($i+1$)-th iteration by $q^{(i)}({\varPhi_{k}^{(-\theta)}})$ under total $L$ iterations \cite{bishop2006pattern}. 

We next specify the likelihood function (\ref{eq3}) to obtain the first proposed filter (Algorithm 1), and then specify the likelihood function (\ref{eq4})  to obtain the second proposed filter (Algorithm 2).

\subsection{Proposed filters}

\subsubsection{Measurement update based on MtG assumption }

Rewrite (\ref{eq15}) as 
\begin{small}
\begin{align}\label{eq16}
&p\left(\mathbf{z}_{1: k} | \varPhi_{k}\right) p\left(\varPhi_{k} | \varPsi_{k}\right) =\text{N}({\bf{x}}_{k}; {\bf{{\hat{x}}}}_{k|k-1}, {\bf{{{P}}}}_{k|k-1} ) \operatorname{IG}(\sigma_{k} ; \hat{\alpha}_{k | k-1}, \hat{\beta}_{k | k-1})\notag\\
&\quad \quad \quad \quad \times p(\mathbf{z}_{1: k-1}) \mathrm{N}(\mathbf{z}_{k} ; \bar{m}_{k}{\bf{H}}{\bf{x}}_{k}, \tilde{\mathbf{R}}_{k}) \mathrm{N}\left(\mathbf{z}_{k} ; \bar{m}_{k} \mathbf{H} \mathbf{x}_{k}, \check{\mathbf{R}}_{k} \right).
\end{align}
\end{small}
First, we let $\theta=\mathbf{x}_{k}$ and substitute (\ref{eq16}) into (\ref{eq14}). It then follows that
 \begin{align*}
\log q^{(i+1)}\left(\mathbf{x}_{k}\right) &\propto\mathrm{N}(\mathbf{z}_{k} ; \bar{m}_{k}{\bf{H}}{\bf{x}}_{k}, \tilde{\mathbf{R}}_{k}^{(i)}) \mathrm{N}(\mathbf{z}_{k} ; \bar{m}_{k}{\bf{H}}{\bf{x}}_{k}, \check{\mathbf{R}}_{k}^{(i)})\\
&\times \mathrm{N}(\mathbf{x}_{k} ; \hat{\mathbf{x}}_{k \mid k-1}, {\mathbf{P}}_{k \mid k-1}^{(i)}).
 \end{align*}
Therefore, $q^{(i+1)}(\mathbf{x}_{k})$ can be approximated using 
$q^{(i+1)}(\mathbf{x}_{k})=\mathrm{N}(\mathbf{x}_{k} ; \hat{\mathbf{x}}_{k \mid k}^{(i+1)},\mathbf{P}_{k \mid k}^{(i+1)})$,
where $\hat{\mathbf{x}}_{k | k}^{(i+1)}$ and $\mathbf{P}_{k | k}^{(i+1)}$ are given by
\begin{scriptsize}
	\begin{align}
&{\bar{\mathbf{R}}_{k}}^{(i+1)}= {\bf{\sigma}}_{k}^{(i)}{\bf{H}}{\bf{{{S}}}}_{k}^{(i)}{\bf{H}}^{\top}+{\bar{m}_{k}^2}{\bf{H}}{\bf{{{P}}}}_{k|k-1}{\bf{H}}^{\top}+{\bf{R}}_{k},  \label{eq17}\\
&\mathbf{K}_{k}^{(i+1)}=\bar{m}_k{\mathbf{P}}_{k | k-1} \mathbf{H}^{\top}[{\bar{\mathbf{R}}_{k}}^{(i+1)}]^{-1},\\
&\hat{\mathbf{x}}_{k | k}^{(i+1)}=\hat{\mathbf{x}}_{k | k-1}+\mathbf{K}_{k}^{(i+1)}(\mathbf{z}_{k}-\bar{m}_k\mathbf{H} \hat{\mathbf{x}}_{k | k-1}),\\
&\mathbf{P}_{k | k}^{(i+1)}={\mathbf{P}}_{k | k-1}-\bar{m}_k\mathbf{K}_{k}^{(i+1)} \mathbf{H} {\mathbf{P}}_{k | k-1},\\
&{\bf{{{S}}}}_{k}^{(i+1)}\approx\text{E}[ {\bf{{\hat{x}}}}_{k|k}^{(i+1)}  ({\bf{{\hat{x}}}}_{k|k}^{(i+1)})^{\top}], \label{eq21}
\end{align}
\end{scriptsize}
where $ \hat{\mathbf{x}}_{k | k-1}$ and  ${\mathbf{P}}_{k | k-1}$ are given in (\ref{eq12})-(\ref{eq13}).

Then, letting $\theta=\sigma_{k}$ and substituting (\ref{eq16}) into (\ref{eq14}) yields 
$\log q^{(i+1)}\left(\sigma_{k}\right) \propto\mathrm{N}(\mathbf{z}_{k} ; \bar{m}_{k}{\bf{H}}{\bf{x}}_{k}, \tilde{\mathbf{R}}_{k}^{(i)})
\operatorname{IG}(\sigma_{k} ; \hat{\alpha}_{k | k}, \hat{\beta}_{k | k})$, we have 
$q^{(i+1)}(\sigma_{k})=\operatorname{IG}(\sigma_{k}; \hat{\alpha}_{k|k}^{(i+1)}, \hat{\beta}_{k|k}^{(i+1)})$, where 
\begin{scriptsize}
\begin{align}
&\hat{\alpha}_{k|k}^{(i+1)}=\hat{\alpha}_{k|k-1}+\frac{1}{2},\label{eq22}\\
&\hat{\beta}_{k|k}^{(i+1)}=\hat{\beta}_{k|k-1}+\frac{1}{2} \text{tr}  \label{eq23}\\
&\resizebox{\hsize}{!}{$\left[ (\mathbf{z}_{k}-\bar{m}_{k}{\bf{H}}\hat{\mathbf{x}}_{k | k}^{(i+1)})({\bf{H}}{\bf{{{S}}}}_{k}^{(i+1)}{\bf{H}}^{\top})^{-1}(\mathbf{z}_{k}-\bar{m}_{k}{\bf{H}}\hat{\mathbf{x}}_{k | k}^{(i+1)})^{\top}+\bar{m}_k^2\mathbf{H} {\mathbf{P}}_{k | k}^{(i+1)}\mathbf{H}^{\top}\right],$}\notag\\
&\hat{\sigma}_{k}^{(i+1)}={\hat{\beta}_{k|k}^{(i+1)}}/{\hat{\alpha}_{k|k}^{(i+1)}}. \label{eq24}
\end{align}
\end{scriptsize}

The proposed MtG filter is summarized as Algorithm 1, where $L$ denotes the iteration number and $\eta$ denotes the threshold.   
\begin{algorithm}[htb] 
	\caption{Proposed MtG filter} 
	\label{alg:Framwork} 
	  \begin{algorithmic}[1] 
		\Require 
		$\mathbf{F}$, $\mathbf{H}$, 
		${\mathbf{Q}}_{k}$, ${\mathbf{R}}_{k}$,  $\bar{m}_k$,
		$\hat{\mathbf{x}}_{k-1 | k-1}$, $\mathbf{P}_{k-1 | k-1}$,  $\mathbf{z}_{k}$, $\rho$, $\hat{\alpha}_0$, $\hat{\beta}_0$, $\eta$, $L$;
		\State Initialize $\hat{\alpha}_{k | k-1}$  and 
	$\hat{\beta}_{k | k-1}$ as in (\ref{eq7})-(\ref{eq8});  
		\label{code:fram:extract} 
		\State Update  ${\bf{{\hat{x}}}}_{k|k-1} $
		and ${\bf{{{P}}}}_{k|k-1}$ as in (\ref{eq12})-(\ref{eq13}); \\ ${\hat{\alpha}_{k|k}^{(0)}}=\hat{\alpha}_{k|k-1}+1/2$, ${\hat{\beta}_{k|k}^{(0)}}=\rho\hat{\beta}_{k|k-1}$,\\ $\hat{\sigma}_{k}^{(0)}=\frac{\hat{\beta}_{k|k}^{(0)}}{\hat{\alpha}_{k|k}^{(0)}}$, ${\bf{{{S}}}}_{k}^{(0)}\approx\text{E}[ {\bf{{\hat{x}}}}_{k|k-1} {\bf{{\hat{x}}}}_{k|k-1}^{\top}]$; 
		\label{code:fram:trainbase} 
		\State  \textbf{for} $i=0:L-1$  \textbf{do} 
		\State \quad Update $\hat{\mathbf{x}}_{k | k}^{(i+1)}$, $\mathbf{P}_{k | k}^{(i+1)}$, and ${\bf{{{S}}}}_{k}^{(i+1)}$ as in (\ref{eq17})-(\ref{eq21}); 
		  \label{code:fram:select}
		  \State  \quad  
		  If $\frac{\|\hat{\mathbf{x}}_{k \mid k}^{(i+1)}-\hat{\mathbf{x}}_{k \mid k}^{(i)}\|}{\|\hat{\mathbf{x}}_{k \mid k}^{(i)}\|} \leq \eta$, stop iteration;
		\State  \quad Update  $\hat{\alpha}_{k|k}^{(i+1)}$, 
		${\hat{\beta}_{k|k}^{(i+1)}}$ and $\hat{\sigma}_{k}^{(i+1)}$
		as in (\ref{eq22})-(\ref{eq24}); 
		\label{code:fram:classify} 
		\State  \textbf{end for}  \\
		\Return \\$\hat{\mathbf{x}}_{k \mid k}=\hat{\mathbf{x}}_{k \mid k}^{(L)}, \mathbf{P}_{k \mid k}=\mathbf{P}_{k \mid k}^{(L)}, \hat{\alpha}_{k|k}=\hat{\alpha}_{k|k}^{(L)}, \text { and } \hat{{\beta}}_{k \mid k}=\hat{{\beta}}_{k \mid k}^{(L)}$; 
	    \Ensure 
		$\hat{\mathbf{x}}_{k \mid k}, \mathbf{P}_{k \mid k}, \hat{\alpha}_{k|k},  \text { and } \hat{{\beta}}_{k \mid k}$; 
	\end{algorithmic} 
\end{algorithm}
\subsubsection{Measurement update based on Std assumption}
We see from the above section that the MtG filter estimates multiplicative noise covariance directly. However, the difference between multiplicative ($m_{k}$) and additive (${\bf{v}}_{k}$)  measurement noise actually cannot be identified in reality. Therefore, we take them as a whole to estimate. 
In this case, as mentioned above, the measurement likelihood is modelled as Std (\ref{eq4}). Note that the proposed Std filter is different from the existing Std filters \cite{huang2017novel1}, where both process and measurement equations are modelled and heavy-tailed additive noise is considered. On the other hand, the proposed Student's $\emph{t}$ filter in this section is applied to model the non-Gaussian property of likelihood caused by multiplicative noise. 

Similar to the MtG filter, we rewrite (\ref{eq15}) as 
\begin{align}\label{eq119}
&p\left(\mathbf{z}_{1: k} | \varPhi_{k}\right) p\left(\varPhi_{k} | \varPsi_{k}\right) \notag\\
&=p(\mathbf{z}_{k} | \mathbf{x}_{k}, \bar{\mathbf{R}}_{k}, \lambda_{k}) p(\mathbf{x}_{k} | \mathbf{z}_{1: k-1}) p(\bar{\mathbf{R}}_{k} | \mathbf{z}_{1: k-1})  p(\lambda_{k}) p(\mathbf{z}_{1: k-1}) \notag \\
&=\mathrm{N}(\mathbf{z}_{k} ; \bar{m}_{k}{\bf{H}}{\bf{x}}_{k}, \bar{\mathbf{R}}_{k} / \lambda_{k}) \mathrm{G}(\lambda_{k} ; \frac{v}{2}, \frac{v}{2}) \mathrm{N}(\mathbf{x}_{k} ; \hat{\mathbf{x}}_{k | k-1}, \mathbf{P}_{k | k-1}) 
\notag \\ 
&\quad \times \operatorname{IW}(\bar{\mathbf{R}}_{k} ; \hat{u}_{k | k-1}, \hat{\mathbf{U}}_{k | k-1}) p\left(\mathbf{z}_{1: k-1}\right).
\end{align}
Then, $\log p(\varPhi_{k}, \mathbf{z}_{1: k})$ can be formulated as 
\begin{align}\label{eq26}
&\log p(\varPhi_{k}, \mathbf{z}_{1: k})=(\frac{m+v}{2}-1) \log \lambda_{k}-\frac{v}{2} \lambda_{k}\notag\\
&\quad -\frac{1}{2}(m+\hat{u}_{k | k-1}+2) \log |\bar{\mathbf{R}}{_{k}}|\notag\\
&\quad -\frac{1}{2}\lambda_{k} (\mathbf{z}_{k}-\bar{m}_{k}\mathbf{H} \mathbf{x}_{k})^{\top} \bar{\mathbf{R}}_{k}^{-1}(\mathbf{z}_{k}-\bar{m}_{k}\mathbf{H} \mathbf{x}_{k})\notag\\
&\quad -\frac{1}{2} \operatorname{tr}(\hat{\mathbf{U}}_{k | k-1} \bar{\mathbf{R}}_{k}^{-1})\notag\\\
&\quad -\frac{1}{2}(\mathbf{x}_{k}-\hat{\mathbf{x}}_{k | k-1})^{\top } \mathbf{P}_{k | k-1}^{-1}(\mathbf{x}_{k}-\hat{\mathbf{x}}_{k | k-1})+C.
\end{align}
First, letting  $\theta=\lambda_{k}$ and substituting (\ref{eq26}) into (\ref{eq14}), we obtain
\begin{scriptsize}
\begin{equation}
\log q^{(i+1)}\left(\lambda_{k}\right)=(\frac{m+v}{2}-1) \log \lambda_{k} \notag -\frac{1}{2}\{v+\operatorname{tr}(\mathbf{E}_{k}^{(i)} \bar{\mathbf{R}}_{k}^{-1})\} \lambda_{k},
\end{equation}
\end{scriptsize}
where
\begin{align}\label{eq27}
\mathbf{E}_{k}^{(i)}&=\mathrm{E}^{(i)}[(\mathbf{z}_{k}-\bar{m}_k\mathbf{H} \mathbf{x}_{k})(\mathbf{z}_{k}-\bar{m}_k\mathbf{H} \mathbf{x}_{k})^{\top}]\\
&=(\mathbf{z}_{k}-\bar{m}_{k}\mathbf{H} \hat{\mathbf{x}}_{k | k}^{(i)})(\mathbf{z}_{k}-\bar{m}_{k}\mathbf{H}\hat{\mathbf{x}}_{k | k}^{(i)})^{\top}+\bar{m}_{k}^2\mathbf{H} \mathbf{P}_{k | k}^{(i)} \mathbf{H}^{\top}.\notag
\end{align}
Then, $ q^{(i+1)}\left(\lambda_{k}\right)$ can be updated as a Gamma distribution with shape parameter $\gamma_{k}^{(i+1)}$ and rate parameter $\delta_{k}^{(i+1)}$ being
\begin{align}
\gamma_{k}^{(i+1)}&=\frac{1}{2}(m+v), \label{eq28}\\
\delta_{k}^{(i+1)}&=\frac{1}{2}\{v+\operatorname{tr}(\mathbf{E}_{k}^{(i)} [\bar{\mathbf{R}}{_{k}}]^{-1})\}.\label{eq29}
\end{align}
Second, letting  $\theta=\bar{\mathbf{R}}_{k}$ and using (\ref{eq26}) in (\ref{eq14}), we obtain 
\begin{align}
\begin{aligned}
&\log q^{(i+1)}(\bar{\mathbf{R}}_{k})
=-\frac{1}{2}(m+\hat{u}_{k | k-1}+2) \log |\bar{\mathbf{R}}_{k}| \\
&\quad \quad \quad -\frac{1}{2} \mathrm{tr} (\mathbf{B}_{k}^{(i)}\mathrm{E}^{(i+1)}\left[\lambda_{k}\right]+\hat{\mathbf{U}}_{k | k-1} (\bar{\mathbf{R}}_{k})^{-1})+C,
\end{aligned}
\end{align}
where $\mathbf{B}_{k}^{(i)}$ is
\begin{equation}\label{eq31}
\begin{aligned}
\mathbf{B}_{k}^{(i)}&= \mathrm{E}^{(i)}[\left(\mathbf{z}_{k}-m_k\mathbf{H} \mathbf{x}_{k}\right)\left(\mathbf{z}_{k}-m_k\mathbf{H} \mathbf{x}_{k}\right)^{\top}] \\
=&(\mathbf{z}_{k}-\bar{m}_{k}\mathbf{H} \hat{\mathbf{x}}_{k | k}^{(i)})(\mathbf{z}_{k}-\bar{m}_{k}\mathbf{H}\hat{\mathbf{x}}_{k | k}^{(i)})^{\top}+\bar{m}_{k}^2\mathbf{H} \mathbf{P}_{k | k}^{(i)} \mathbf{H}^{\top}.
\end{aligned}
\end{equation}
Employing (31), $\bar{\mathbf{R}}_{k}$ is updated as $q^{(i+1)}(\bar{\mathbf{R}}_{k})=\operatorname{IW}(\bar{\mathbf{R}}_{k} ; \hat{u}_{k}^{(i+1)}, \hat{\mathbf{U}}_{k}^{(i+1)})$ with shape parameters being 
\begin{align}
\hat{u}_{k}^{(i+1)}&=\hat{u}_{k | k-1}+1, \label{eq32}\\
\hat{\mathbf{U}}_{k}^{(i+1)}&=\mathbf{B}_{k}^{(i)}\mathrm{E}^{(i+1)}\left[\lambda_{k}\right]+\hat{\mathbf{U}}_{k | k-1}.\label{eq33}
\end{align}
Third, letting  $\theta=\mathbf{x}_{k}$ and using (\ref{eq26}) into (\ref{eq14}), we obtain
\begin{align}
\log q^{(i+1)}\left(\mathbf{x}_{k}\right)&=-0.5\mathrm{E}^{(i+1)}\left[\lambda_{k}\right] (\mathbf{z}_{k}-\bar{m}_{k}\mathbf{H} \mathbf{x}_{k})^{\top}\notag\\ \times&\mathrm{E}^{(i+1)}[\bar{\mathbf{R}}_{k}^{-1}](\mathbf{z}_{k}-\bar{m}_{k}\mathbf{H} \mathbf{x}_{k})-0.5(\mathbf{x}_{k}-\hat{\mathbf{x}}_{k | k-1})^{\top}\notag\\
\times& \mathbf{P}_{k | k-1}^{-1}(\mathbf{x}_{k}-\hat{\mathbf{x}}_{k | k-1})+C,
\end{align}
where
$\mathrm{E}^{(i+1)}\left[\lambda_{k}\right]=\frac{\gamma_{k}^{(i+1)}}{\delta_{k}^{(i+1)}}$ and  $\mathrm{E}^{(i+1)}[{\bar{\mathbf{R}}_{k}^{-1}}]$  are 
\begin{align}
\mathrm{E}^{(i+1)}\left[\lambda_{k}\right]&={\gamma_{k}^{(i+1)}}/{\delta_{k}^{(i+1)}},\\
\mathrm{E}^{(i+1)}[{\bar{\mathbf{R}}_{k}^{-1}}]&=(\hat{u}_{k}^{(i+1)}-n-1)(\hat{\mathbf{U}}_{k}^{(i+1)})^{-1}. \label{eq36}
\end{align}
Define the modified measurement likelihood 
\begin{align}
p^{(i+1)}(\mathbf{z}_{k} | \mathbf{x}_{k}) &=\mathrm{N}(\mathbf{z}_{k} ; \bar{m}_k\mathbf{H} \mathbf{x}_{k}, \bar{\mathbf{R}}_{k}^{(i+1)}),\label{eq37}
\end{align}
where  ${\bar{\mathbf{R}}_{k}}^{(i+1)}$ is formulated as 
\begin{equation}
 {\bar{\mathbf{R}}_{k}}^{(i+1)}={\{\mathrm{E}^{(i+1)}[\bar{\mathbf{R}}_{k}^{-1}]\}^{-1}}/{\mathrm{E}^{(i+1)}\left[\lambda_{k}\right]}.\label{eq38}
\end{equation}
Then,  we obtain $q^{(i+1)}(\mathbf{x}_{k})=\mathrm{N}(\mathbf{x}_{k} ; \hat{\mathbf{x}}_{k | k}^{(i+1)}, \mathbf{P}_{k | k}^{(i+1)}),$
where $\hat{\mathbf{x}}_{k | k}^{(i+1)}$, $\mathbf{P}_{k | k}^{(i+1)}$ and gain matrix $\mathbf{K}_{k}^{(i+1)}$ are given by
\begin{scriptsize}
\begin{align}
&\mathbf{K}_{k}^{(i+1)}=\bar{m}_k{\mathbf{P}}_{k \mid k-1} \mathbf{H}^{\top}[{\bar{\mathbf{R}}_{k}}^{(i+1)}]^{-1},\label{eq39}\\
&\hat{\mathbf{x}}_{k | k}^{(i+1)}=\hat{\mathbf{x}}_{k | k-1}+\mathbf{K}_{k}^{(i+1)}(\mathbf{z}_{k}-\bar{m}_k\mathbf{H} \hat{\mathbf{x}}_{k | k-1}),\\
&\mathbf{P}_{k | k}^{(i+1)}={\mathbf{P}}_{k \mid k-1}-\bar{m}_k\mathbf{K}_{k}^{(i+1)} \mathbf{H} {\mathbf{P}}_{k \mid k-1}.\label{eq41}
\end{align}
\end{scriptsize}
Finally, after $L$-step iteration, the posterior $q(\lambda_{k})$, $q(\bar{\mathbf{R}}_{k}) $, and $q(\mathbf{x}_{k})$ can be approximated as 
\begin{scriptsize}
\begin{align}
&q(\lambda_{k}) \approx q^{(L)}(\lambda_{k})=\mathbf{G}(\lambda_{k} ; \gamma_{k}^{(L)}, \delta_{k}^{(L)}), \notag\\
&q(\bar{\mathbf{R}}_{k})=\operatorname{IW}(\bar{\mathbf{R}}_{k} ; \hat{u}_{k | k}, \hat{\mathbf{U}}_{k | k}) \approx q^{(L)}(\bar{\mathbf{R}}_{k})=\operatorname{IW}(\bar{\mathbf{R}}_{k} ; \hat{u}_{k}^{(L)}, \hat{\mathbf{U}}_{k}^{(L)}), \notag\\
&q(\mathbf{x}_{k})=\mathrm{N}(\mathbf{x}_{k} ; \hat{\mathbf{x}}_{k | k}, \mathbf{P}_{k | k})\approx q^{(L)}(\mathbf{x}_{k}) =\mathrm{N}(\mathbf{x}_{k} ; \hat{\mathbf{x}}_{k | k}^{(L)}, \mathbf{P}_{k | k}^{(L)}) . \notag
\end{align}
\end{scriptsize}
The proposed Student's $\emph{t}$ filter is summarized in Algorithm 2.
\begin{algorithm}
	\caption{Proposed Student's $\emph{t}$ filter}
		\begin{algorithmic}[1] 
	\Require {$\mathbf{F}$, $\mathbf{H}$, $\mathbf{z}_{k}$, $\bar{m}_k$, ${\mathbf{Q}}_{k-1}$, $\hat{u}_{k-1 | k-1}$, $\hat{\mathbf{U}}_{k-1 | k-1}$, $\hat{\mathbf{x}}_{k-1 | k-1}$, $\mathbf{P}_{k-1 | k-1}$,  $m$, $n$, $\rho$, $L$, $\eta$, $c_0$, $d_0$;}
	\State\text{Update} $
	\hat{\mathbf{x}}_{k | k-1}$ and ${\mathbf{P}}_{k | k-1}$ as in (\ref{eq12})-(\ref{eq13});
	\State \text{Initialize} 
	\begin{align*}
		&\hat{\mathbf{x}}_{k | k}^{(0)}=\hat{\mathbf{x}}_{k | k-1}, \mathbf{P}_{k | k}^{(0)}={\mathbf{P}}_{k | k-1}, \\
		&\hat{u}_{k | k-1}=\rho(\hat{u}_{k-1 | k-1}-m-1)+m+1, \\
		&\hat{\mathbf{U}}_{k | k-1}=\rho \hat{\mathbf{U}}_{k-1 | k-1},\\
		&\mathrm{E}^{(0)}[{\bar{\mathbf{R}}_{k}^{-1}}]=(\hat{u}_{k | k-1}-n-1)\hat{\mathbf{U}}_{k | k-1}^{-1}; 
	\end{align*}

	\State \textbf{for} $i=0:L-1$ \textbf{do} 	
	\State  \quad Update
		$q^{(i+1)}(\lambda_{k})$: 
		$\mathbf{E}_{k}^{(i)}$,
		$\gamma_{k}^{(i+1)}$, and $\delta_{k}^{(i+1)}$ as in (\ref{eq27})-(\ref{eq29}),\\  \quad  \quad  \quad  \quad \quad  \quad \quad  \quad  \quad  $\mathrm{E}^{(i+1)}\left[\lambda_{k}\right]=\gamma_{k}^{(i+1)} / \delta_{k}^{(i+1)}$;
      \State  \quad Update	$q^{(i+1)}(\bar{\mathbf{R}}_{k})$:
		$\mathbf{B}_{k}^{(i)}$,
		$\hat{u}_{k}^{(i+1)}$, $\hat{\mathbf{U}}_{k}^{(i+1)}$ as in (\ref{eq31})-(\ref{eq33});\\
		 \quad Update $q^{(i+1)}(\mathbf{x}_{k})$:
		$\mathrm{E}^{(i+1)}[\bar{\mathbf{R}}_{k}^{-1}]$,
		${\bar{\mathbf{R}}_{k}}^{(i+1)}$ as in (\ref{eq36}) and (\ref{eq38}); \\
		\quad \quad  \quad \quad  \quad  \quad  \quad \quad $\mathbf{K}_{k}^{(i+1)}$, 
		$\hat{\mathbf{x}}_{k | k}^{(i+1)}$, and 
		$\mathbf{P}_{k | k}^{(i+1)}$ as in (\ref{eq39})-(\ref{eq41});
		\State  \quad  
		If ${\|\hat{\mathbf{x}}_{k \mid k}^{(i+1)}-\hat{\mathbf{x}}_{k \mid k}^{(i)}\|}/{\|\hat{\mathbf{x}}_{k \mid k}^{(i)}\|} \leq \eta$, stop iteration;
    	\State  \textbf{end for} \\
	    \Return \begin{align*}
	    \hat{\mathbf{x}}_{k | k}&=\hat{\mathbf{x}}_{k | k}^{(L)}, \mathbf{P}_{k | k}=\mathbf{P}_{k | k}^{(L)}, \gamma_{k}=\gamma_{k}^{(L)}, \delta_{k}=\delta_{k}^{(L)},\\ \hat{u}_{k | k}&=\hat{u}_{k}^{(L)}, \hat{\mathbf{U}}_{k | k}=\hat{\mathbf{U}}_{k}^{(L)};	\end{align*}
    	\Ensure  {$\hat{\mathbf{x}}_{k | k}$, $\mathbf{P}_{k | k}$, $\gamma_{k}$, $\delta_{k}$, $\hat{u}_{k | k}$, $\hat{\mathbf{U}}_{k | k}$;}
		\end{algorithmic} 
\end{algorithm}

\begin{remark}
\emph{The current state-of-the-art filters are proposed either for known multiplicative noise or only for additive time-varying noise covariance \cite{huang2017novel,sarkka2009recursive}.  Since multiplicative noise leads the likelihood function to be a non-Gaussian distribution, the existing filters are not applicable to the case of multiplicative noise. The proposed filters can well depict the non-Gaussian property of likelihood, so as to estimate the states and noise covariance more accurately. }
\end{remark}

\section{Performance analysis}\label{sw4}

In this section, we first analyze the influence of the chosen parameters on the convergence of the proposed filters. Then, we calculate the upper and lower bounds of the estimation error covariance.  Finally, we analyze the convergence of VB inference using a fixed-point iteration.

\subsection{Parameter influence}

In this section, we analyze the influence of the chosen covariance parameters on the convergence of the proposed filters.

\subsubsection{Convergence analysis for the cases of different parameters}

The proposed two filters characterize the distribution of noise covariance as 
 $p(\sigma_{k} ) =\operatorname{IG}(\sigma_{k}; \hat{\alpha}_{k | k}, \hat{\beta}_{k | k})$ and $p({\bar{\bf{R}}}_{k}) =\operatorname{IW}({\bar{\bf{R}}}_{k} ; \hat{u}_{k | k}, \hat{\mathbf{U}}_{k | k})$, respectively. We next analyze the 
 influence of these parameters on covariance estimation. 
 
 \textbf{\emph{(1)  Convergence analysis of $\hat{\alpha}_{k | k}$ and $\hat{u}_{k | k}$. }} 

 We first show the convergence of the sequence  $\{\hat{\alpha}_{k | k}^i\}$ and $\rho \hat{\alpha}_{k|k}^i \gg 1 / 2$ when $0.6<\rho<1$.

According to equations (\ref{eq7}) and (\ref{eq22}), sequence $\{\hat{\alpha}_{k | k}^i\}$ can be obtained from the following recursive form
\begin{align}
\hat{\alpha}_{k+1|k+1}^i &=1 / 2+\rho \hat{\alpha}_{k|k}^i, \tag{a1} \label{a1}\\
\hat{\alpha}_{k+2|k+2}^i &=1 / 2+\rho \hat{\alpha}_{k+1|k+1}^i. \tag{a2} \label{a2}
\end{align}
Subtracting (\ref{a2}) from (\ref{a1}) yields
\begin{equation}
\hat{\alpha}_{k+2|k+2}^i-	\hat{\alpha}_{k+1|k+1}^i=\rho(\hat{\alpha}_{k+1|k+1}^i-\rho \hat{\alpha}_{k|k}^i). \tag{a3} \label{a3}
\end{equation}
Let $b_{k|k}^i=\hat{\alpha}_{k+1|k+1}^i-\rho \hat{\alpha}_{k|k}^i$. Then, (\ref{a3}) can be rewritten as 
\begin{equation}
b_{k+1|k+1}^i=\rho b_{k|k}^i,\tag{a4} \label{a4}
\end{equation}
where $b_{1|1}^i=\hat{\alpha}_{2|2}^i-\rho \hat{\alpha}_{1|1}^i=1/2+\rho \hat{\alpha}_{1|1}^i-\hat{\alpha}_{1|1}^i $.
Therefore,  sequence $\{b_{k | k}^i\}$ can be organized as 
\begin{equation}
b_{k|k}^i=\rho^{k-1} b_{1|1}^i. \tag{a5} \label{a5}
\end{equation}
From (\ref{a3}), it follows that 
\begin{equation}
\hat{\alpha}_{k+1|k+1}^i-\hat{\alpha}_{k|k}^i=\rho^{k-1} b_{1|1}^i. \tag{a6} \label{a6}
\end{equation}
This  yields
\begin{align}
&\hat{\alpha}_{k+1|k+1}^i-\hat{\alpha}_{k|k}^i=\rho^{k-1} b_{1|1}^i, \tag{a6}\\
&\hat{\alpha}_{k|k}^i-\hat{\alpha}_{k-1|k-1}^i=\rho^{k-2} b_{1|1}^i, \tag{a7} \label{a7}\\
& \quad \quad \cdots \notag\\
&\hat{\alpha}_{2|2}^i-\hat{\alpha}_{1|1}^i= b_{1|1}^i. \tag{a8} \label{a8}
\end{align}
By summing from (\ref{a6}) to (\ref{a8}), we obtain 
\begin{equation}
\hat{\alpha}_{k+1|k+1}^i-\hat{\alpha}_{1|1}^i=
b_{1|1}^i \sum_{j=0}^{k-1} \rho^{j} =\frac{b_{1|1}^i(1-\rho^{k})}{1-\rho}. \tag{a9} \label{a9}
\end{equation}
If $0.6<\rho<1$, it follows that 
\begin{align}
\begin{aligned}
\lim _{k \rightarrow+\infty} \hat{\alpha}_{k+1|k+1}^i &=\frac{b_{1|1}^i}{1-\rho}+\hat{\alpha}_{1|1}^i \\
&=\frac{1 / 2+\rho \hat{\alpha}_{1|1}^i-\hat{\alpha}_{1|1}^i}{1-\rho}+\hat{\alpha}_{1|1}^i \\
&=\frac{1 / 2}{1-\rho}.
\end{aligned} \tag{a10} \label{a10}
\end{align}
Thus, $\{\hat{\alpha}_{k | k}^i\}$  converges to $\frac{1 / 2}{1-\rho}$. In addition, (\ref{a10}) indicates  that $\rho \hat{\alpha}_{k|k}^i \gg 1 / 2$.

Similarly, for sequence $\{\hat{u}_{k | k}^i\}$,  we can obtain 
\begin{align}
\lim _{k \rightarrow+\infty} \hat{u}_{k+1|k+1}^i=\frac{m+2-\rho(m+1)}{1-\rho},
 \tag{a11} \label{a11}
\end{align}
which proves that $\{\hat{u}_{k | k}^i\}$  converges to  $\frac{m+2-\rho(m+1)}{1-\rho}$.

\textbf{\emph{(2)  Convergence analysis of $\hat{\beta}_{k | k}$, $\hat{\mathbf{U}}_{k | k}$, $\hat{\sigma}_{k}$, and ${\bar{\mathbf{R}}_{k}}$.}}

From equations (\ref{eq8}) and (\ref{eq23}), we obtain 
\begin{equation}
\hat{\beta}_{k|k}^{(i+1)}=\rho{\hat{\beta}}_{k-1|k-1}^{(i+1)}+\frac{1}{2} \text{tr}(\mathbf{A}_k),  \tag{a12} \label{a12}
\end{equation}
where 
\begin{equation*}
\resizebox{\hsize}{!}{$\mathbf{A}_k=\left[ (\mathbf{z}_{k}-\bar{m}_{k}{\bf{H}}\hat{\mathbf{x}}_{k | k}^{(i+1)})({\bf{H}}{\bf{{{S}}}}_{k}^{(i+1)}{\bf{H}}^{\top})^{-1}(\mathbf{z}_{k}-\bar{m}_{k}{\bf{H}}\hat{\mathbf{x}}_{k | k}^{(i+1)})^{\top}+(\bar{m}_k)^2\mathbf{H} {\mathbf{P}}_{k | k}^{(i+1)}\mathbf{H}^{\top}\right].$}
\end{equation*}
From equations (\ref{eq10}) and (\ref{eq33}), we can obtain 
\begin{equation}
\hat{\mathbf{U}}_{k}^{(i+1)}=\mathbf{B}_{k}^{(i)}\mathrm{E}^{(i+1)}\left[\lambda_{k}\right]+\rho \hat{\mathbf{U}}_{k-1}^{(i+1)}. \tag{a13} \label{a13}
\end{equation}
From (\ref{a12}) and (\ref{a13}), we know that  $\hat{\beta}_{k|k}^{(i+1)}$ and  $\hat{\mathbf{U}}_{k}^{(i+1)}$ are estimated by using the state estimation and its error covariance at each iteration, which is crucial for the estimation performance of noise covariances thereby.  $\hat{\beta}_{k|k}^{(i+1)}$ and  $\hat{\mathbf{U}}_{k}^{(i+1)}$ directly affect the covariance estimations  $\hat{\sigma}_{k}$  and  $\bar{\bf{R}}_{k}$, respectively. Meanwhile, they are conversely affected by the state estimation.
Therefore, this establishes the relation between state estimation and noise covariance estimation.

In order to analyze the convergence of  $\hat{\beta}_{k | k}$, we notice that a time-varying matrix $\mathbf{A}_k$ is in (\ref{a12}). The following theorem indicates that $\text{tr}(\mathbf{A}_k)$ is bounded.	
\begin{theorem}[The bounds of $\text{tr}(\mathbf{A}_k)$]
\emph{	If Assumptions 2 and 3  in \cite{yu2021robust}  are satisfied, and Lemma 1 and Theorem 3  hold, 
then	$\text{tr}(\mathbf{A}_k)$ has a uniform upper bound and a uniform lower bound, i.e.,
	\begin{small}
		\begin{equation*}
		\resizebox{\hsize}{!}{$\frac{\bar{m}_{k}^2 n  \alpha_{1} }{(1+n^{2} \beta_{1} \beta_{2})n} \text{tr}(\mathbf{H}\mathbf{H}^{\top}) \leq \text{tr}(\mathbf{A}_k) \leq \frac{n\bar{m}_{k}^2 (1+n^{2} \beta_{1} \beta_{2})}{\alpha_{2}} \text{tr}(\mathbf{H}\mathbf{H}^{\top}) + \alpha  \text{tr}(\mathbf{H}{\bf{{{S}}}}_{k}\mathbf{H}^\top)^{-1},$}
		\end{equation*}
	\end{small}}
where $\alpha_{1}$, $\alpha_{2}$, $\beta_{1}$, and $\beta_{2}$  are defined in Theorem 3, and $\alpha$ is defined in Lemma 1.  
\end{theorem}
\begin{proof}
\emph{Please refer to \cite{yu2021robust} for details.}
\end{proof}

For the convergence of the sequence $\{\hat{\beta}_{k | k}^{(i+1)}\}$, it can be proved using a similar approach in  \cite{yu2021robust}.

Furthermore, from equation (\ref{eq24})  ($\hat{\sigma}_{k}^{(i+1)}={{\hat{\beta}}_{k|k}^{(i+1)}}/{\hat{\alpha}_{k|k}^{(i+1)}}$), we can obtain that the multiplicative noise covariance sequence $\{\hat{\sigma}_{k}^{(i+1)}\}$ is bounded and convergent. 

We next show the convergence of the sequence 
$\{{\hat{\mathbf{U}}}_{k}^{(i+1)}\}$. 

From (\ref{a13}), we notice that $\mathrm{E}^{(i+1)}[\lambda_{k}]$ and $\mathbf{B}_{k}^{(i)}$ affect the convergence of sequence 
$\{{{\hat{\mathbf{U}}}}_{k|k}^{(i+1)}\}$. $\mathrm{E}^{(i+1)}[\lambda_{k}]$ is shown to converge to a bounded value in the following statements. For simplicity, we ignore $\mathrm{E}^{(i+1)}\left[\lambda_{k}\right]$ and rewrite (\ref{a13}) as 
\begin{equation*}
\hat{\mathbf{U}}_{k}^{(i+1)}=\mathbf{B}_{k}^{(i)}+\rho \hat{\mathbf{U}}_{k-1}^{(i+1)}. 
\end{equation*}
We first set $\tilde{\mathbf{B}}_{k}=\mathbf{B}_{k+1}^{(i)}-\mathbf{B}_{k}^{(i)}$ and $\tilde{\mathbf{U}}_{k}=\hat{\mathbf{U}}_{k+1}^{(i+1)}-\hat{\mathbf{U}}_{k}^{(i+1)}$. Then, by using the same method with sequence $\{{\hat{\beta}}_{k | k-1}^{(i+1)}\}$, we can obtain  
\begin{small}
\begin{equation*}
\hat{\mathbf{U}}_{k}^i-\hat{\mathbf{U}}_{1}^i=\tilde{\mathbf{U}}_{1} \sum\limits_{j=0}^{k-1} \rho^{j}+\frac{1}{2} \tilde{\mathbf{B}}_2 \sum\limits_{j=1}^{k-2}\rho^{j}+\cdots+\frac{1}{2} \tilde{\mathbf{B}}_{k-1} \sum\limits_{j=0}^{1}\rho^{j}+\frac{1}{2} \tilde{\mathbf{B}}_k.
\end{equation*}
\end{small}
If $0<\rho<1$, it follows that 
\begin{small}
\begin{align*}
\lim_{k \rightarrow+\infty} {\hat{\mathbf{U}}_{k}^i} &=
\frac{\tilde{\mathbf{U}}_{1}}{1-\rho}+{\hat{\mathbf{U}}_{1}^i}+\frac{ \tilde{\mathbf{B}}_2} {1-\rho}+...+{\hat{\mathbf{U}}_{1}^i}+ \tilde{\mathbf{B}}_{k-1} ({1+\rho})+\tilde{\mathbf{B}}_{k}\\
&=\frac{(\rho-1)\hat{\mathbf{U}}_{1}^i+ \text{tr}({\mathbf{B}}_1)}{1-\rho}+{\hat{\mathbf{U}}_{1}^i}+\frac{ (\tilde{\mathbf{B}}_2+\tilde{\mathbf{B}}_3+\cdots)}{1-\rho}
\\
&=\frac{ ({\mathbf{B}}_1+\tilde{\mathbf{B}}_2+\cdots)}{1-\rho}\\
&=\frac{{\mathbf{B}}_{\infty}}{1-\rho}.
\end{align*}
\end{small}
According to equation (\ref{eq31}), we know that ${\mathbf{B}}_{k}$ is positive definite. Therefore,   sequence $\{{\hat{\mathbf{U}}}_{k}^{(i+1)}\}$ is convergent. 

Furthermore, from equation (\ref{eq36}) ($\mathrm{E}^{(i+1)}[{\bar{\mathbf{R}}_{k}^{-1}}]=(\hat{u}_{k}^{(i+1)}-n-1)(\hat{\mathbf{U}}_{k}^{(i+1)})^{-1}$), we can obtain that the total measurement noise covariance sequence $\{\bar{\mathbf{R}}_{k}\}$ is bounded and convergent. 

\textbf{\emph{(3) Convergence analysis of $\gamma_{k}$, $\lambda_{k}$, and $\delta_{k}$.}}

First, since $m$ and $v$ are constants, $\gamma_{k}^{i+1}$ is always a constant, which indicates the convergence of $\gamma_{k}^{i+1}$. 

Second, $\bar{\mathbf{R}}_{k}$ is positive definite (given in Section A of Part 5). According to the definition of $\mathbf{E}_{k}^{(i)}$ in (\ref{eq27}), we obtain that
$\mathbf{E}_{k}^{(i)}$ is positive definite.  

Then, we use the trace approach similar to $\text{tr}(\mathbf{A}_k)$ and we can obtain that 
 $\operatorname{tr}(\mathbf{E}_{k}^{(i)} [\bar{\mathbf{R}}_{k}]^{-1})$ is bounded. 
 
Furthermore, since $\delta_{k}^{i+1}=\frac{1}{2}\{v+\operatorname{tr}(\mathbf{E}_{k}^{(i)} [\bar{\mathbf{R}}{_{k}}]^{-1})\}$, by using the similar approach on  $\{\hat{\alpha}_{k | k}^i\}$, we can obtain that $\delta_{k}^{i+1}$ convergences to a bounded value.

Finally, since $\mathrm{E}^{(i+1)}\left[\lambda_{k}\right]={\gamma_{k}^{(i+1)}}/{\delta_{k}^{(i+1)}}$, we can easily prove that $\lambda_{k}$ is convergent and bounded.

\subsubsection{Relation between state estimation and noise covariance estimation
}

The following theorem explains the relation between state estimation error and  noise covariance estimation. 
\begin{theorem}
\emph{Let $\tilde{\mathbf{x}}_{k | k}=\hat{\mathbf{x}}_{k | k}-\hat{\mathbf{x}}_{k | k}^{*}$ be the estimation error between the proposed state estimation $\hat{\mathbf{x}}_{k | k}$ and the optimal state estimation $\hat{\mathbf{x}}_{k | k}^{*}$, and $\tilde{\check{\bf{R}}}=\check{\bf{R}}-{\bar{\bf{R}}}$ be the noise derivation between the real noise covariance ($\check{\bf{R}}$) and the estimation (${\bar{\bf{R}}}$). If the innovation  ($\mathbf{z}_{k}-\bar{m}_k\mathbf{H} \hat{\mathbf{x}}_{k | k-1}$) does not change dramatically, the state estimation error $\tilde{\mathbf{x}}_{k | k}$ and the noise estimation derivation $\tilde{\check{\bf{R}}}$ are positively interrelated. In particular,  
\begin{equation}
\resizebox{\hsize}{!}{$\tilde{\mathbf{x}}_{k | k}=\bar{m}_k{\mathbf{P}}_{k | k-1} \mathbf{H}^{\top}\left(\mathbf{U}^{-1}(\mathbf{U}^{-1}+\tilde{\check{\bf{R}}}^{-1})^{-1}{\mathbf{U}}^{-1}\right)(\mathbf{z}_{k}-\bar{m}_k\mathbf{H} \hat{\mathbf{x}}_{k | k-1})$}, \tag{b1} \label{b1}
\end{equation}
where $\mathbf{U}=\bar{m}_k^2\mathbf{H} {\mathbf{P}}_{k | k-1}\mathbf{H}^{\top}+{\bar{\bf{R}}}$. This means that the accuracy of state estimation increases with the improvement in estimation performance of noise covariance, and vice versa.}
\end{theorem}

\begin{proof}
\emph{Please refer to \cite{yu2021robust} for details.}	
\end{proof}

\subsubsection{Performance with different $\rho$}\label{ss4}

Then, we discuss the influence of parameter $\rho$ (forgetting factor). 
From equations (\ref{eq36}) and (\ref{eq38}), we obtain 
\begin{equation*}\label{eq42}
{\bar{\mathbf{R}}_{k}}^{(i+1)}=\frac{\hat{\mathbf{U}}_{k}^{(i+1)}}{(\hat{u}_{k}^{(i+1)}-m-1){\mathrm{E}^{(i+1)}\left[\lambda_{k}\right]}}.
\end{equation*}
Based on equations (\ref{eq32})-(\ref{eq33}), the above equation can be rewritten as 
\begin{scriptsize}
\begin{equation*}
{\bar{\mathbf{R}}_{k}}^{(i+1)}=\frac{\rho \hat{\mathbf{U}}_{k-1 | k-1}+\mathbf{B}_{k}^{(i)}}{\rho(\hat{u}_{k-1 | k-1}-m-1)+1}.
\end{equation*}
\end{scriptsize}
According to $p({\bar{\bf{R}}}_{k} | \mathbf{z}_{1: k-1}) =\operatorname{IW}({\bar{\bf{R}}}_{k} ; \hat{u}_{k | k-1}, \hat{\mathbf{U}}_{k | k-1})$, ${\bar{\mathbf{R}}_{k-1}}^{(i+1)}$ is formulated as 
$
{\bar{\mathbf{R}}_{k-1}}^{(i+1)}=\frac{\hat{\mathbf{U}}_{k-1 | k-1}}{\hat{u}_{k-1 | k-1}-m-1}.$

Combining the above equations yields 
\begin{equation*}
{\bar{\mathbf{R}}_{k}}^{(i+1)}=\frac{\mathbf{B}_{k}^{(i)}+\rho(\hat{u}_{k-1 | k-1}-m-1) {(\bar{\mathbf{R}}_{k-1})}}{1+\rho(\hat{u}_{k-1| k-1} -m-1)}.
\end{equation*}
Based on equations (\ref{eq9}) and (\ref{eq32}), we obtain 
\begin{equation*}
\hat{u}_{k | k}=\rho(\hat{u}_{k-1 | k-1}-m-1)+m+2.
\end{equation*} 
Then, it follows that
\begin{equation*}
\hat{u}_{k-1 | k-1}-m-1=\rho^{k-1}(\hat{u}_{0 | 0}-m-1)+(1-\rho^{k-1}) /(1-\rho).
\end{equation*}
Set 
$\varsigma(k,\rho)=\rho(\hat{u}_{k-1 | k-1}-m-1)$ and we obtain
\begin{equation}
{\bar{\mathbf{R}}_{k}}^{(i+1)}=\frac{\varsigma(k,\rho) {\bar{\mathbf{R}}_{k-1}}+\mathbf{B}_{k}^{(i)}}{\varsigma(k,\rho)+1}.  \tag{c1} \label{c1}
\end{equation}
Then, $\varsigma(k,\rho)$ can be rewritten as 
\begin{equation}
\varsigma(k,\rho)=\rho^{k}(\hat{u}_{0|0}-m-1)+(\rho-\rho^{k}) /(1-\rho). \tag{c2} \label{c2}
\end{equation}
This yields 
\begin{equation}
\lim _{k \rightarrow+\infty}\varsigma(k,\rho)=\rho /(1-\rho). \tag{c3} \label{c3}
\end{equation}
From (\ref{c1})-(\ref{c3}), we obtain that $\varsigma(k,\rho)$ is a tradeoff parameter to balance the weight of ${\bar{\mathbf{R}}_{k-1}}$ and $\mathbf{B}_{k}^{(i)}$. In addition, $\varsigma(k,\rho)$ is monotonically increasing with respect to $\rho$. Therefore, $\rho\in(0,1]$  can be directly applied to balance the time update ${\bar{\mathbf{R}}_{k-1}}$ and measurement update ${\bar{\mathbf{R}}_{k}}^{(i+1)}$. Considering that multiplicative noise is
always widely slow-varying in engineering applications \cite{huang2017novel,sarkka2009recursive}, this paper sets the range as  $\rho \in [0.6, 1]$. Similarly, we can analyze the IG case.

\subsubsection{Performance with  $\bar{\mathbf{R}}_{0}$}\label{ss4}

We next discuss the effect of $\tilde{\mathbf{R}}_{0}$ (${\sigma}_{0}$ can be similarly available).
In the fixed-point iteration, the initial value is
\begin{align}
{\bar{\mathbf{R}}_{k}}^{(0)}={\bar{\mathbf{R}}_{k-1}}. \tag{d1} \label{d1}
\end{align}
Set $\tau_k=\frac{\zeta(k,\rho) }{\zeta(k,\rho)+1}$, $\mathbf{C}_k=\frac{\mathbf{B}_{k}^{(L-1)}}{\zeta(k,\rho)+1}$,
where $0<\tau_{k}<1$ and $\mathbf{C}_{k}\geq \bf{0}$, (\ref{c1}) is rewritten as 
\begin{equation}
{\bar{\mathbf{R}}_{k}}=\tau_k{\bar{\mathbf{R}}_{k-1}}+\mathbf{C}_{k}. \tag{d2} \label{d2}
\end{equation}
It then follows that 
\begin{equation}
{\bar{\mathbf{R}}_{k-1}}=\sum\nolimits_{i=1}^{k-1}(\prod\nolimits_{j=i+1}^{k-1} \tau_{j}) \mathbf{C}_{i}+(\prod\nolimits_{i=1}^{k-1} \tau_{i}) {\bar{\mathbf{R}}_{0}}. \tag{d3} \label{d3}
\end{equation}
From (\ref{d3}), we know that the initial value has effects on  ${\bar{\mathbf{R}}_{k}}^{(0)}$. ${\bar{\mathbf{R}}_{0}}$ proportionally affects on $\bar{\mathbf{R}}_{k}$ as $k$ increases. To ensure that 
 ${\bar{\mathbf{R}}_{k}}^{(i)}$ converges to the true value, we need to choose
appropriate initial value ${\bar{\mathbf{R}}_{k}}^{(0)}$ to ensure local convergence of VB inference. Therefore, these parameters need to be selected near their true values and we rely on our engineering experience. Thanks to the diagonal form, we can select   $\bar{\mathbf{R}}_{0}$ as $\bar{\mathbf{R}}_{0}=\text{daig}[r_{1},\cdots,r_{j},\cdots,r_{m}]$, where $r_{j}>0$.

\subsubsection{Performance with  ${\bar{\mathbf{R}}_{k}}$}\label{ss4}

Finally, we prove the stability of the presented filter with an initial value $\bar{\mathbf{R}}_{0}$. Based on (\ref{d3}), $\bar{\mathbf{R}}_{0}>\mathbf{0}$, $0<\tau_{j}<1$ and $\mathbf{C}_{j}\geq \mathbf{0}$, we obtain ${\bar{\mathbf{R}}_{k-1}}>\mathbf{0}$.
Note that (\ref{eq31}) yields $\mathbf{B}_{k}^{(i)} \geq \mathbf{0}$.
Then, substituting the above equations into (\ref{d3}) yields 
${\bar{\mathbf{R}}_{k}}^{(i+1)}>\mathbf{0}$.
This means that ${\bar{\mathbf{R}}_{k}}$ is positive definite and indicates that the proposed filter is stable.

\subsection{The bounds of the error covariance}

This section calculates the upper and lower bounds of the estimation error covariance.

In order to obtain the error bounds, it is convenient to define the equivalent system of (\ref{eq1})-(\ref{eq2}) as 
	\begin{align}
	{\bf{x}}_{k} &= {\bf{F}}{\bf{x}}_{k-1} + {\bf{w}}_{k-1}, \notag\\
	{\bf{z}}_{k} &={\bf{H}}_{k}^m{\bf{x}}_{k} + {\bf{v}}_{k}^m, \tag{g1}  \label{g1} 
	\end{align}
	where the state equation is the same as (\ref{eq1}), ${\bf{H}}_{k}^m$ is defined as ${\bf{H}}_{k}^m\triangleq \bar{m}_{k}{\bf{H}}_{k}$($|\bar{m}_{k}|<\infty$),  and ${\bf{v}}_{k}^m$ is a zero-mean Gaussian white noise and satisfies 
	$ \text{E}[ ({\bf{v}}_{j}^m)({\bf{v}}_{k}^m)^{\top}]={\bf{R}}^e_{k}\delta_{kj}$, where ${\bf{R}}^e_{k}= {\bf{\sigma}}_{k}{\bf{H}}_{k}{\bf{{{S}}}}_{k}{\bf{H}}^{\top}_{k}+{\bf{R}}_{k}$. ${\bf{{{S}}}}_{k}$ is defined as ${\bf{S}}_{k} \triangleq \text{E}[ {\bf{x}}_{k}{\bf{x}}^{\top}_{k}]={\bf{F}}{\bf{S}}_{k-1}{\bf{F}}^\top+{\bf{Q}}_{k-1}$ with a boundary condition ${\bf{S}}_{0}=\mathbf{P}_{0 | 0}$.  Since ${\bf{R}}_{k}$ is positive definite, so is ${\bf{R}}^e_{k}$. Then, we define
	\begin{align}
	M(T, T-l)&=\sum\nolimits_{k=T-l}^{T} {\bf{F}}_{k}^{\top} ({\bf{H}}_{k}^m)^{\top} ({\bf{R}}^e_{k})^{-1} ({\bf{H}}_{k}^m) {\bf{F}}_{k}, \tag{g2} \label{g2} \\
	W(T, T-l)&=\sum\nolimits_{k=T-l}^{T-1} {\bf{F}}_{k+1} {\bf{Q}}_{k}  {\bf{F}}_{k+1}^{\top}, \tag{g3}  \label{g3} 
	\end{align}
	where
	${\bf{F}}={\bf{F}}_{k}$ (${{\bf{F}}_{k}}{\bf{F}}_{k-1}={\bf{F}}_{k,k-1}$ and 
	${\bf{F}}_{T,k}={\bf{F}}_{T, T-1} {\bf{F}}_{T-1, T-2} \cdots {\bf{F}}_{k+1, k}$ ($k \leq T $)). Equations (\ref{g2}) and (\ref{g3}) are the stochastic observability matrix and stochastic controllability matrix, respectively. The following theorem on the bounds of the error covariance matrix is established. 
\begin{theorem}
\emph{If the system (\ref{eq1}) and (\ref{g1}) is uniformly controllable and uniformly observable, i.e.,
\begin{align}
\alpha_1 I \leq W(k, k-N+1) \leq \beta_1 I,  \tag{g4} \label{g4} \\
\alpha_2 I \leq M(k, k-N+1) \leq \beta_2 I,   \tag{g5}  \label{g5} 
\end{align}
where ( $\alpha_1$,  $\alpha_2$,  $\beta_1$ and $\beta_2>0$) and $N$ is a positive integer, and assume $\mathbf{P}_{0 | 0} >0$, the estimation error covariance has a uniform upper bound and a uniform lower bound, i.e, 
\begin{equation*}
\frac{\alpha_{1}}{1+n^{2} \beta_{1} \beta_{2}} I \leq \mathbf{P}_{k | k} \leq \frac{1+n^{2} \beta_{1} \beta_{2}}{\alpha_{2}} I.
\end{equation*}}
\end{theorem}
\begin{proof}
\emph{Please refer to \cite{yu2021robust}  for details.}
\end{proof}

\subsection{Convergence analysis of VB inference}

This section investigates the convergence of the fixed-point iteration of the iterative VB procedure (\ref{eq14}), which is used in the derivation process. Since $\varPhi_{k}=\{{\bf{x}}_{k},{\bf{\sigma}}_{k}, {\bar{\bf{R}}}_{k}\}$ is considered, the proposed filters take the form 
\begin{equation}
\varPhi^{(i)}=T(\varPhi^{(i-1)}), \tag{e1} \label{e1}  
\end{equation}
for $i=1, \ldots, L$, where $T$ denotes a certain continuous mapping and $\{\varPhi^{(i)}\}$ denotes the sequence of $\varPhi$. Since (\ref{e1}) is difficult to analyze, we take another form 
\begin{equation}
\varPhi^{(i)}=(1-\varepsilon) \varPhi^{(i-1)}+\varepsilon T(\varPhi^{(i-1)}) \triangleq \Omega_{n}^{\varepsilon}(\varPhi^{(i-1)}), \tag{e2} \label{e2}  
\end{equation}
where $\varepsilon>0$ (when $\varepsilon=1$, (\ref{e2}) becomes (\ref{e1})) and $\Omega_{n}^{\varepsilon}$ is a continuous mapping with respect to $\varepsilon$ and the time index $n$.

We use the parameter set $\varTheta =\left\{{\bf{\sigma}}_{k}, {\bar{\bf{R}}}_{k}\right\}$ and the state ${\bf{x}}_{k}$ in the fixed-point iteration.  Then, the iterative stage in the algorithm corresponding to (16) is 
\begin{small}
\begin{equation}
\varTheta^{(i)}=(1-\varepsilon) \varTheta^{(i-1)}+\varepsilon T(\varTheta^{(i-1)}) \triangleq \Omega_{k}^{\varepsilon}(\varTheta^{(i-1)}). \tag{e3} \label{e3} 
\end{equation}
\end{small}
In the fixed-point iteration of the VB procedure, we apply an approximation probability density $q(\mathbf{x}_{k}, \varTheta)$ for $p(\mathbf{x}_{k}, \varTheta|{\bf{y}}_{k})$, which can be factorized as $q(\mathbf{x}_{k}, \varTheta)=q(\mathbf{x}_{k})q(\varTheta)$. The states $q(\mathbf{x}_{k})$ and the values of the hyperparameters of $q(\varTheta)$ are obtained by an iterative procedure. For the $(i)$-th iteration,  $(i-1)$-th state estimation $q^{(i-1)}(\mathbf{x}_{k})$ is given. 
 Then, we perform the following two steps of VB procedure. 

\textbf {Step 1}. Optimize the hyperparameters in $q^{(i)}(\varTheta)$  for fixed $\{q^{(i-1)}({\mathbf{x}_{k}}), i=1, \ldots, N\}$.

\textbf {Step 2}. Optimize $q^{(i)}({\mathbf{x}_{k}})$
for fixed $q^{(i)}(\varTheta)$, which is calculated by Step 1.

Now, we investigate the convergence of the aforementioned two steps.  Despite the convergence of a fixed-point algorithm is given in \cite{2015Convergence}, the convergence of VB inference using the fixed-point iteration is still an open issue.

Suppose that the true value of the parameter $\varTheta$ is $\varTheta^\ast$, we next investigate whether the iterative algorithm (\ref{e3}) with the aforementioned two steps is convergent. The following theorem can be established.

\begin{theorem}
\emph{With probability 1 as $n$ approaches infinity, the iterative procedure (\ref{e3}) converges locally to the true value $\varTheta^\ast$ whenever $0<\varepsilon<2$, i.e., (\ref{e3}) converges to the true value $\varTheta^\ast$ whenever the starting value is sufficiently near to $\varTheta^\ast$.}
\end{theorem}
\begin{proof}
\emph{Please refer to \cite{yu2021robust} for details.}
\end{proof}

{\section{Simulations}\label{se4}}

In this section, we use a simulation example to verify the presented filters.
A target moves with a constant velocity in 2-D space following 
a motion model given by $\mathbf{x}_{k}=\left[\begin{array}{cc}
	\mathbf{I}_{2} & \Delta t \mathbf{I}_{2} \\
	\mathbf{0} & \mathbf{I}_{2}
\end{array}\right] \mathbf{x}_{k-1}+\mathbf{w}_{k-1}$ and $\mathbf{z}_{k}=m_k \left[\begin{array}{ll}
\mathbf{I}_{2} & \mathbf{0}
\end{array}\right] \mathbf{x}_{k}+\mathbf{v}_{k}$, where $ \mathbf{x}_{k}=[x_{k},y_{k},\dot{x}_{k},\dot{y}_{k}]^{\top}$ denotes the cartesian coordinates and corresponding velocities,  $\mathbf{I}_{2}$ is a 2-D identity matrix and $\Delta t=1s$ is the sampling interval. 
Multiplicative noise $m_k$ is set to be a Gaussian distribution with mean $\bar{m}_k=5.5$ and covariance ${\bf{\sigma}}_{k}=[2+0.05 \cos (\frac{\pi k}{\text{T}})]$, where $\text{T}=\SI{500}{s}$ is the total time. The additive process and measurement noises are set to be Gaussian  with zero means and covariance matrices ${\text{Q}}_{k}=q\left[\begin{array}{cc}
\frac{\Delta t^{3}}{3} \mathbf{I}_{2} & \frac{\Delta t^{2}}{2} \mathbf{I}_{2} \\
\frac{\Delta t^{2}}{2} \mathbf{I}_{2} & \Delta t \mathbf{I}_{2}
\end{array}\right] $  and $\text{R}_{k}= r {\left[\begin{array}{cc}1 & 0.5 \\0.5 & 1\end{array}\right]}$, where $q$ and $r$ are respectively given as $q= \SI{1}{m^{2}/s^{3}}$ and $r=\SI{100}{m^{2}} $.

Because no prior information on multiplicative noise covariance is available, $\sigma_{0}\sim \operatorname{IG}(1, 1)$ was used. 
In addition, the initial nominal covariance is selected as  $\tilde{\mathbf{R}}_{0}=v \mathbf{I}_{2}$. The nominal covariance KF (KF), the robust Student's \emph{t} Kalman filters (RSTKF 1 and RSTKF 2  \cite{huang2017novel1,huang2016robust}), the VB adaptive Kalman filter (VBAKF) \cite{huang2017novel}, the VB particle filter (VBPF) (200 particles) \cite{yu2019estimation} and the optimal Kalman filter for multiplicative noise (OKF) (given true covariance) are compared. In the proposed filters, we set 
$\rho=0.8$, $\alpha_0=\beta_0=1$, $c_{0}=d_{0}=0.1$, $\eta=10^{6}$, and $L=20$. The parameter of initial covariance $v$ is set as $v=3$. Besides, performance metrics are chosen as the root mean square error (RMSE), the average RMSE (ARMSE), the square root of the normalized Frobenius norm (SRNFN) of the measurement noises, and the average SRNFN (ASRNFN). In particular, 
$\text{RMSE} \triangleq \sqrt{\frac{1}{M} \sum\nolimits_{s=1}^{M}\left(\left(x_{k}^{s}-\hat{x}_{k}^{s}\right)^{2}+\left(y_{k}^{s}-\hat{y}_{k}^{s}\right)^{2}\right)}$, $\text{ARMSE} \triangleq \sqrt{\frac{1}{MT}\sum\nolimits_{k=1}^{T} \sum\nolimits_{s=1}^{M}\left(\left(x_{k}^{s}-\hat{x}_{k}^{s}\right)^{2}+\left(y_{k}^{s}-\hat{y}_{k}^{s}\right)^{2}\right)}$, $\text{SRNFN} \triangleq(\frac{1}{m^{2} M} \sum\nolimits_{s=1}^{M}\|\mathbf{\hat{\bar{R}}}_{k \mid k-1}^{s}-\mathbf{\bar{R}}_{k \mid k-1}^{s}\|^{2})^{\frac{1}{4}}$, and $\text { ASRNFN } \triangleq(\frac{1}{m^{2} M T} \sum\nolimits_{k=1}^{T} \sum\nolimits_{s=1}^{M}\|\mathbf{\hat{\bar{R}}}_{k \mid k-1}^{s}-\mathbf{\mathbf{\bar{R}}}_{ k \mid k-1}^{s}\|^{2})^{\frac{1}{4}}$,  where $(x_{k}^{s}, y_{k}^{s})$ and $(\hat{x}_{k}^{s}, \hat{y}_{k}^{s})$ are, respectively, the true and the estimated variables (position or velocity) at the $s$-th Monte Carlo run, 
$\mathbf{\hat{\bar{R}}}_{k \mid k-1}^{s}$ and $\mathbf{\bar{R}}_{ k \mid k-1}^{s}$ denote, respectively, the estimated and the true total measurement noise covariances, and $M = 100$ denotes the total number of Monte Carlo runs. The initial state is given as  
$\mathbf{x}_{0} \sim \mathrm{N}\left(\mathbf{x}_{0}, \mathbf{P}_{0}\right)$ (initialization (sampling) for VBPF),  where $\mathbf{x}_{0}=[100,100,10,10]^{\mathrm{\top}}$  and  $\mathbf{P}_{0}=100\times  \mathbf{I}_{4}$.

\begin{figure}[!t]\centering
		\includegraphics[width=6cm]{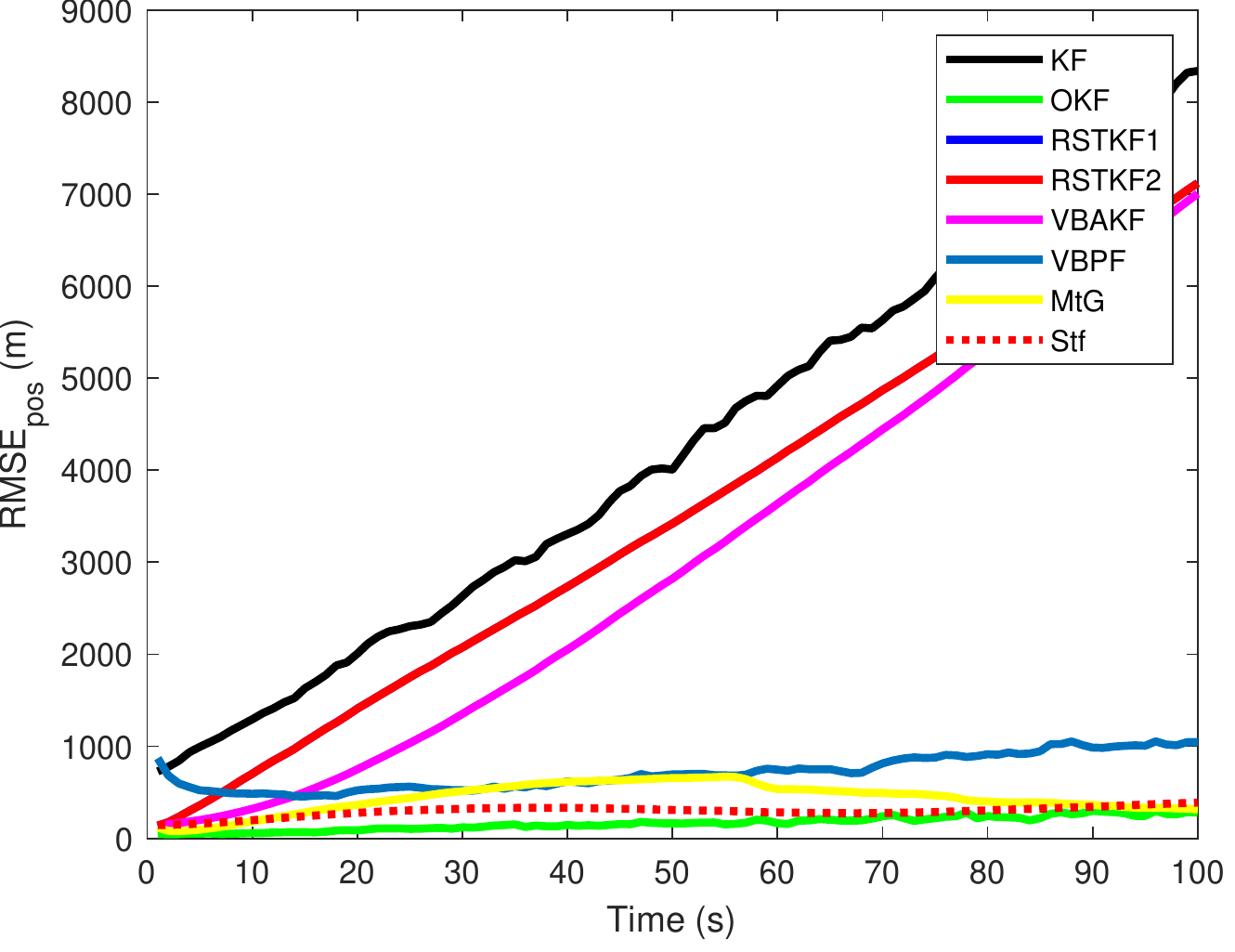}\\
		\includegraphics[width=6cm]{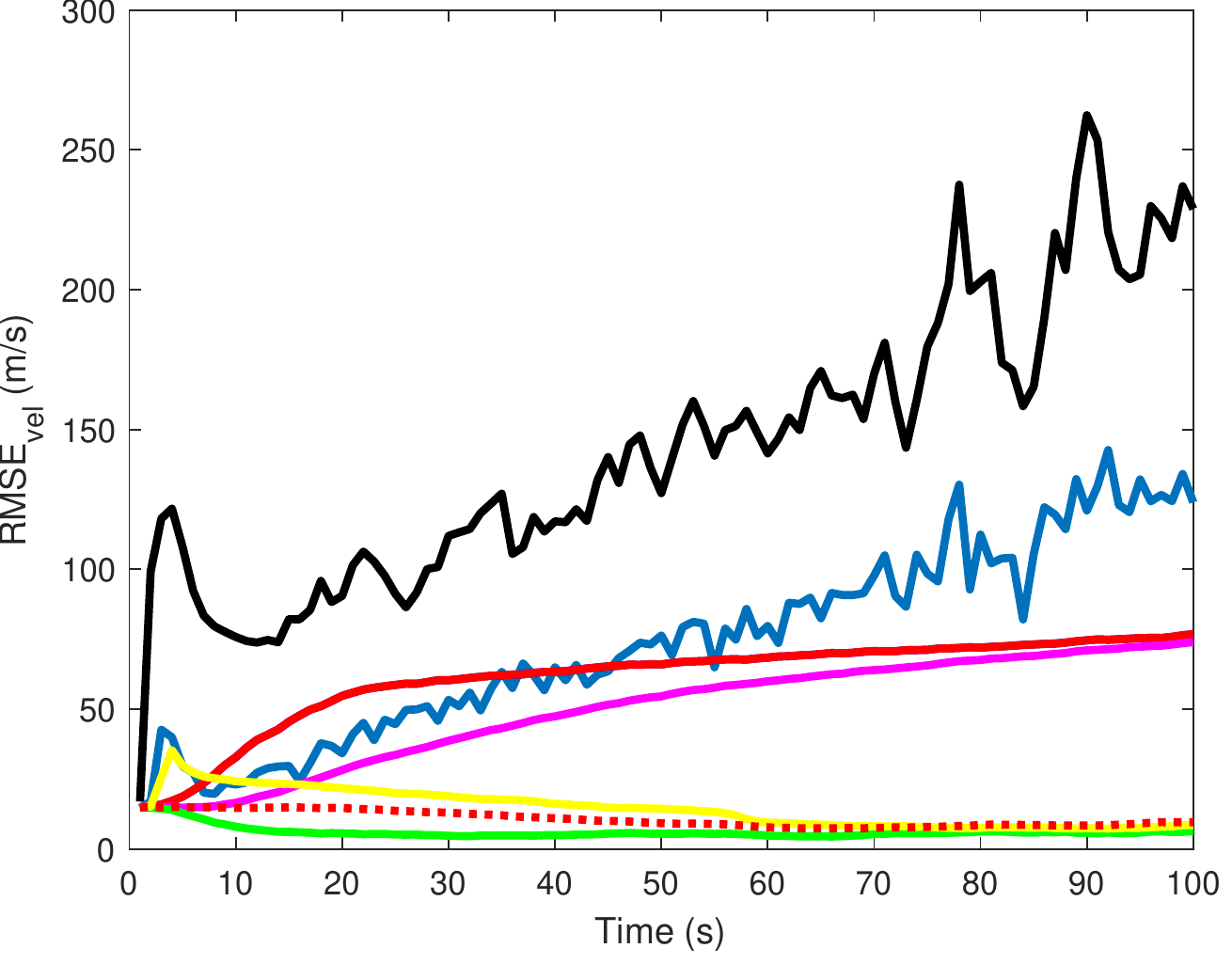}
		\caption{$\text{RMSE}_{\text{pos}}$ and $\text{RMSE}_{\text{vec}}$ of the state-of-the-art filters.}\label{FIG_1}
\end{figure}

\begin{figure}[!t]\centering
		\includegraphics[width=6cm]{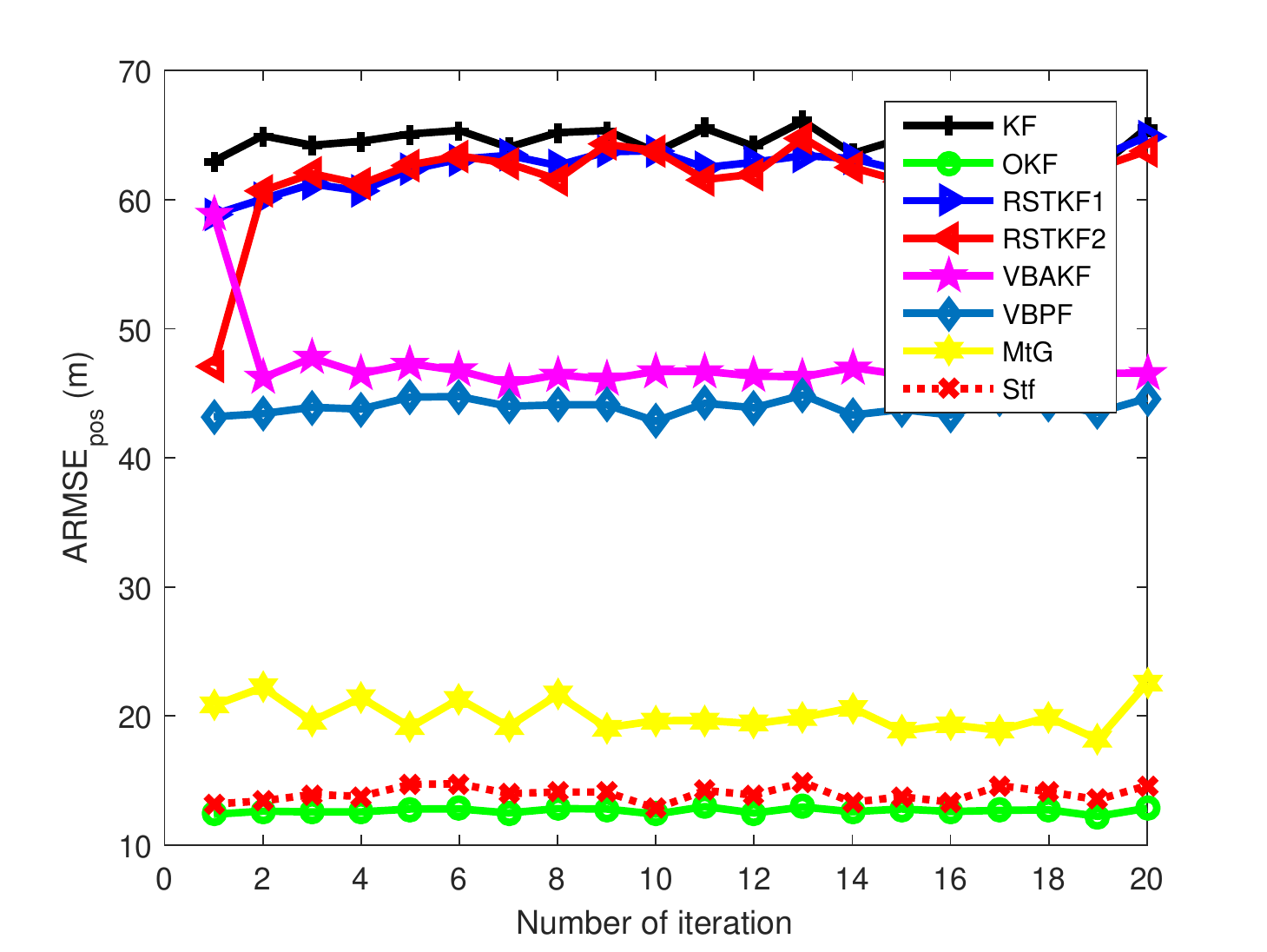}\\
			\includegraphics[width=6cm]{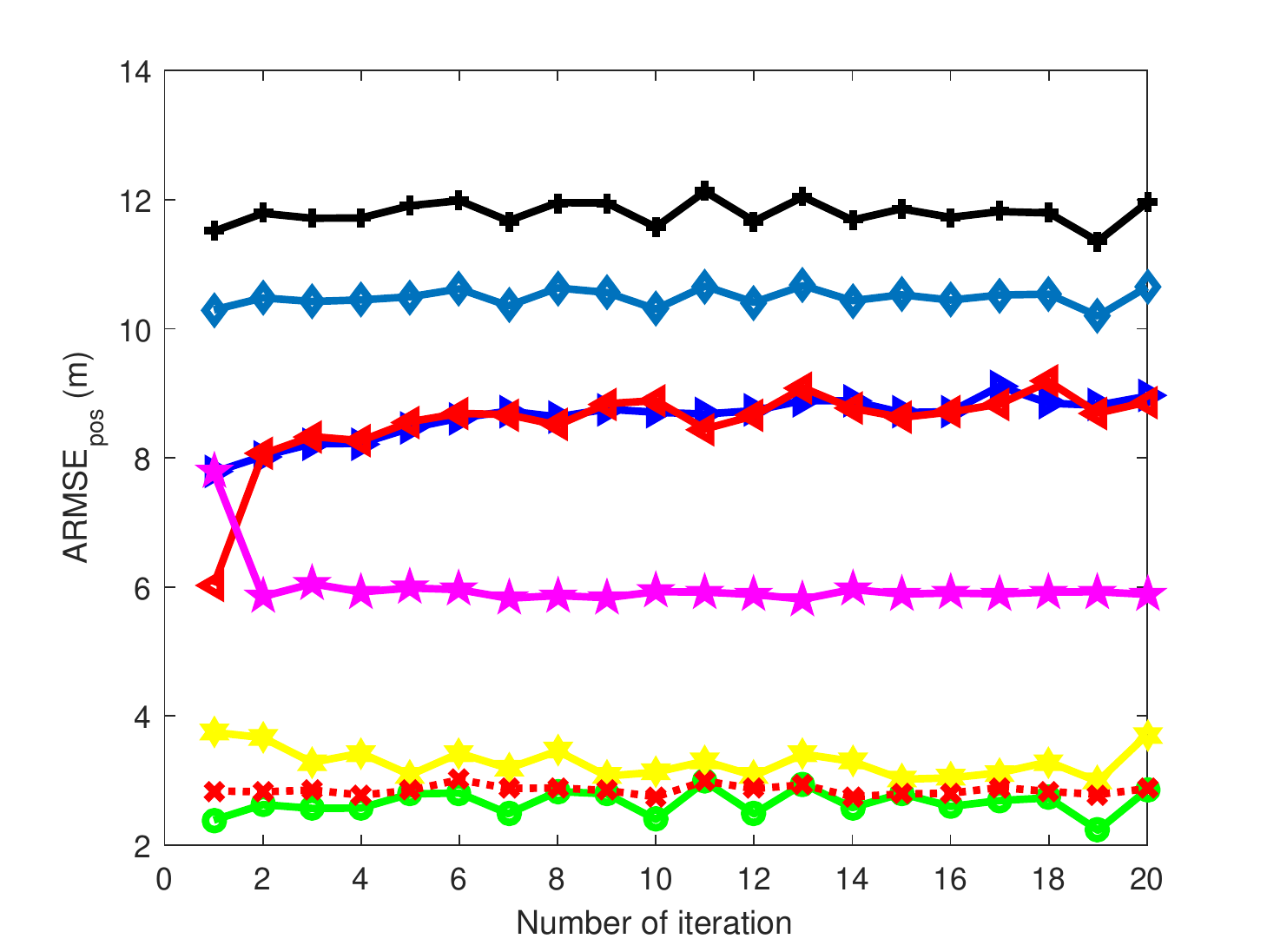}
		\caption{$\text{ARMSE}_{\text{pos}}$ and $\text{ARMSE}_{\text{vec}}$ of the state-of-the-art filters. }\label{FIG_3}
\end{figure}

\begin{figure}[!t]\centering
	\includegraphics[width=6cm]{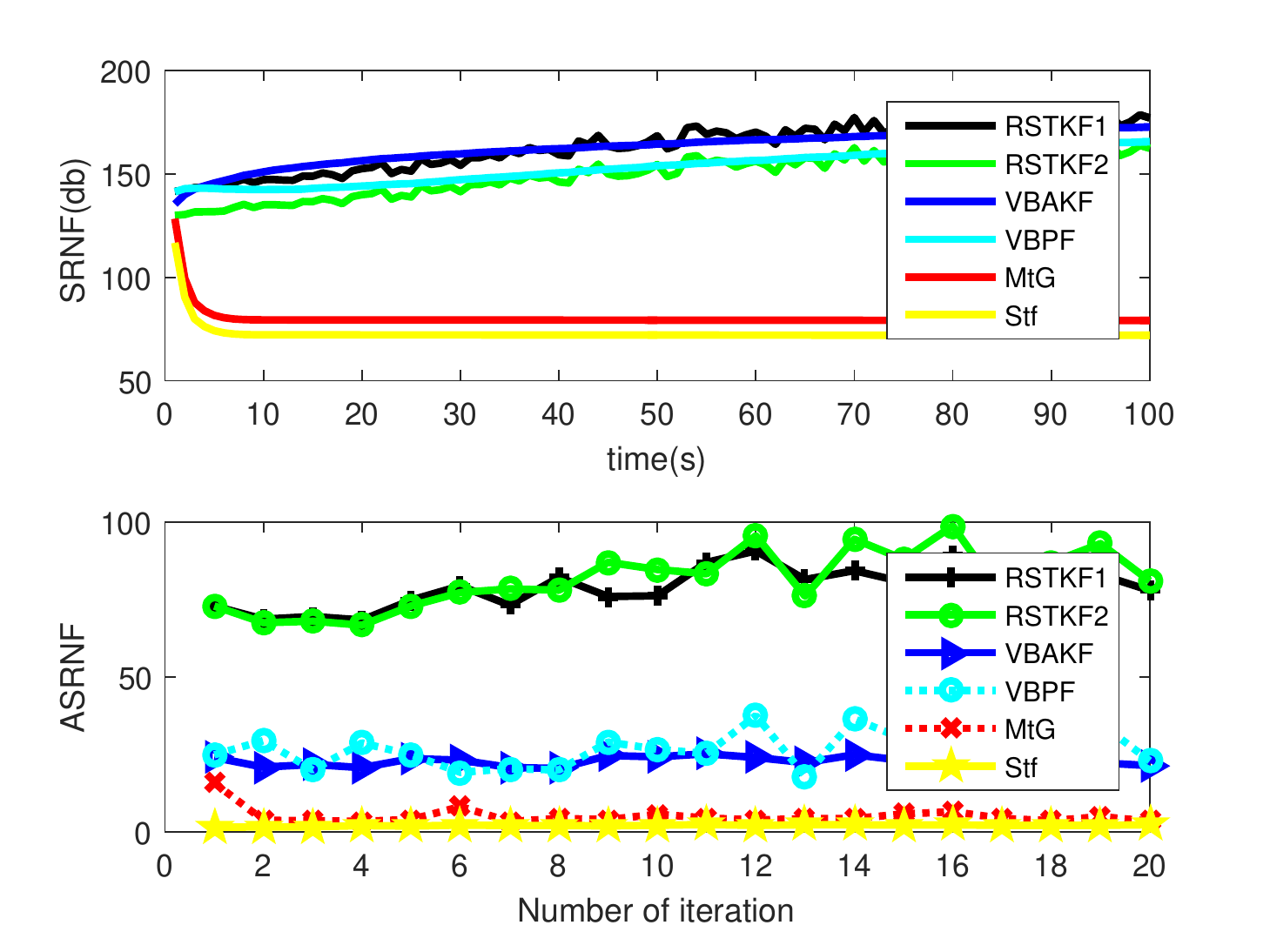}
	\caption{$\text{SRNFN}(\text{log})$ and $\text{ASRNFN} $ of the state-of-the-art filters.}\label{FIG_5}
\end{figure}

The RMSEs  of positions and velocities from the state-of-the-art filters and the presented filters
are shown in Fig. 1. It can be seen that the presented filters have smaller RMSEs than those of state-of-the-art-filters, including VBAKF, VBPF, and RSTKF. It can also be seen that the existing filters diverge eventually, while the proposed filters have robust convergence.
Fig. 2 illustrates the ARMSEs of positions and velocities when $L=[1\quad20]$. It can be seen that the presented filters have smaller ARMSEs and converge to the minimum when $L=2$. 
Fig. 3 shows the quantitative SRNFN and ASRNFN of the existing filters. It is shown that the proposed filters estimate the noise covariance much better than the state-of-the-art filters. Therefore, the proposed filters have a faster convergence rate than those of state-of-the-art filters, including VBAKF, VBPF, and RSTKF.
The Std filter deals with the whole observation noise as a whole, therefore it can get more accurate results. However, the MtG filter only deals with the case of multiplicative noise as an individual one and ignores the possible relationship with the whole. The results are therefore relatively poor but still better than the state-of-the-art filters.

\begin{remark}
\emph{According to the references \cite{huang2017novel,huang2016robust,huang2019novel,yu2021adaptive}, we can see that the model used in the simulation is a typical example of the problem of noise statistics estimation. In fact, the model (\ref{eq1}-\ref{eq2}) considered in this paper is a more generalized form of the model in \cite{huang2017novel,huang2016robust,huang2019novel,yu2021adaptive}. Besides, the existing state-of-the-art filters work either for time-invariant multiplicative noise or only for additive time-varying noise covariances. The proposed filters, on the other hand, are applicable to the case of time-varying multiplicative noise covariance. Since the mean of multiplicative noise in simulation is $5.5$, which indicates that the true state signal is amplified by 5.5 times, the existing filters diverge eventually. However, the proposed filters can effectively	eliminate the influence of multiplicative noise. 
}
\end{remark}

\section{Conclusions} 

In this paper, we studied the joint estimation problem of state and noise covariance for linear systems with unknown covariance of multiplicative noise. Based on assumptions that a Student's $\emph{t}$ distribution and a mixture of two Gaussian distributions as the non-Gaussian likelihood functions, two novel VB based robust filters were developed, where the states together with noise covariances were deduced by choosing the inverse Gamma/Wishart priors. The stability and convergence of the noise covariance parameters and the proposed filters were analyzed. Simulation results illustrated that the presented filters had better performance and were robust enough to resist multiplicative noise.

%\section*{Acknowledgments}

\bibliographystyle{IEEEtran}
\bibliography{yuxingkai}

\begin{thebibliography}{10}
\providecommand{\url}[1]{#1}
\csname url@samestyle\endcsname
\providecommand{\newblock}{\relax}
\providecommand{\bibinfo}[2]{#2}
\providecommand{\BIBentrySTDinterwordspacing}{\spaceskip=0pt\relax}
\providecommand{\BIBentryALTinterwordstretchfactor}{4}
\providecommand{\BIBentryALTinterwordspacing}{\spaceskip=\fontdimen2\font plus
\BIBentryALTinterwordstretchfactor\fontdimen3\font minus
  \fontdimen4\font\relax}
\providecommand{\BIBforeignlanguage}[2]{{%
\expandafter\ifx\csname l@#1\endcsname\relax
\typeout{** WARNING: IEEEtran.bst: No hyphenation pattern has been}%
\typeout{** loaded for the language `#1'. Using the pattern for}%
\typeout{** the default language instead.}%
\else
\language=\csname l@#1\endcsname
\fi
#2}}
\providecommand{\BIBdecl}{\relax}
\BIBdecl

\bibitem{1100100}
R.~{Mehra}, ``Approaches to adaptive filtering,'' \emph{IEEE Transactions on
  Automatic Control}, vol.~17, no.~5, pp. 693--698, 1972.

\bibitem{karasalo2011optimization}
M.~Karasalo and X.~Hu, ``An optimization approach to adaptive {Kalman}
  filtering,'' \emph{Automatica}, vol.~47, no.~8, pp. 1785--1793, 2011.

\bibitem{6179988}
X.~{Gao}, D.~{You}, and S.~{Katayama}, ``Seam tracking monitoring based on
  adaptive {Kalman} filter embedded elman neural network during high-power
  fiber laser welding,'' \emph{IEEE Transactions on Industrial Electronics},
  vol.~59, no.~11, pp. 4315--4325, 2012.

\bibitem{sarkka2009recursive}
S.~S{\"a}rkk{\"a} and A.~Nummenmaa, ``Recursive noise adaptive {Kalman}
  filtering by variational {Bayesian} approximations,'' \emph{IEEE Transactions
  on Automatic control}, vol.~54, no.~3, pp. 596--600, 2009.

\bibitem{schon2011system}
T.~B. Sch{\"o}n, A.~Wills, and B.~Ninness, ``System identification of nonlinear
  state-space models,'' \emph{Automatica}, vol.~47, no.~1, pp. 39--49, 2011.

\bibitem{chen2017maximum}
B.~Chen, X.~Liu, H.~Zhao, and J.~C. Principe, ``Maximum correntropy {Kalman}
  filter,'' \emph{Automatica}, vol.~76, pp. 70--77, 2017.

\bibitem{huang2016robust}
Y.~Huang, Y.~Zhang, N.~Li, and J.~Chambers, ``Robust {Student's} t based
  nonlinear filter and smoother,'' \emph{IEEE Transactions on Aerospace and
  Electronic Systems}, vol.~52, no.~5, pp. 2586--2596, 2016.

\bibitem{huang2017novel1}
Y.~Huang, Y.~Zhang, N.~Li, Z.~Wu, and J.~A. Chambers, ``A novel robust
  {Student's t-Based} {Kalman} filter,'' \emph{IEEE Transactions on Aerospace
  and Electronic Systems}, vol.~53, no.~3, pp. 1545--1554, 2017.

\bibitem{huang2019novel}
Y.~Huang, Y.~Zhang, and J.~A. Chambers, ``A novel {Kullback--Leibler}
  divergence minimization-based adaptive {Student's t-filter},'' \emph{IEEE
  Transactions on Signal Processing}, vol.~67, no.~20, pp. 5417--5432, 2019.

\bibitem{4585346}
V.~{Smidl} and A.~{Quinn}, ``Variational {Bayesian} filtering,'' \emph{IEEE
  Transactions on Signal Processing}, vol.~56, no.~10, pp. 5020--5030, 2008.

\bibitem{ji2006variational}
S.~Ji, B.~Krishnapuram, and L.~Carin, ``Variational {Bayes} for continuous
  hidden markov models and its application to active learning,'' \emph{IEEE
  Transactions on Pattern Analysis and Machine Intelligence}, vol.~28, no.~4,
  pp. 522--532, 2006.

\bibitem{agamennoni2012approximate}
G.~Agamennoni, J.~I. Nieto, and E.~M. Nebot, ``Approximate inference in
  state-space models with heavy-tailed noise,'' \emph{IEEE Transactions on
  Signal Processing}, vol.~60, no.~10, pp. 5024--5037, 2012.

\bibitem{huang2017novel}
Y.~Huang, Y.~Zhang, Z.~Wu, N.~Li, and J.~Chambers, ``A novel adaptive {Kalman}
  filter with inaccurate process and measurement noise covariance matrices,''
  \emph{IEEE Transactions on Automatic Control}, vol.~63, no.~2, pp. 594--601,
  2017.

\bibitem{zhu2021novel}
H.~Zhu, G.~Zhang, Y.~Li, and H.~Leung, ``A novel robust {Kalman} filter with
  unknown non-stationary heavy-tailed noise,'' \emph{Automatica}, vol. 127, p.
  109511, 2021.

\bibitem{wang2020novel}
M.~Wang, Z.~Wang, H.~Dong, and Q.-L. Han, ``A novel framework for
  backstepping-based control of discrete-time strict-feedback nonlinear systems
  with multiplicative noises,'' \emph{IEEE Transactions on Automatic Control},
  vol.~66, no.~4, pp. 1484--1496, 2020.

\bibitem{gao2019robust}
Y.~Gao and Z.~Deng, ``Robust weighted fusion kalman estimators for networked
  multisensor mixed uncertain systems with random one-step sensor delays,
  uncertain-variance multiplicative, and additive white noises,'' \emph{IEEE
  Sensors Journal}, vol.~19, no.~22, pp. 10\,935--10\,946, 2019.

\bibitem{yu2019estimation}
X.~Yu, J.~Li, and J.~Xu, ``Estimation algorithm for system with {non-Gaussian}
  multiplicative/additive noises based on variational {Bayesian} inference,''
  \emph{International Journal of Adaptive Control and Signal Processing},
  vol.~33, no.~4, pp. 586--608, 2019.

\bibitem{liu2020quadratic}
Q.~Liu, Z.~Wang, Q.-L. Han, and C.~Jiang, ``Quadratic estimation for discrete
  time-varying {non-Gaussian} systems with multiplicative noises and
  quantization effects,'' \emph{Automatica}, vol. 113, pp. 1--9, 2020.

\bibitem{springer1966distribution}
M.~D. Springer and W.~Thompson, ``The distribution of products of independent
  random variables,'' \emph{SIAM Journal on Applied Mathematics}, vol.~14,
  no.~3, pp. 511--526, 1966.

\bibitem{bromiley2003products}
P.~Bromiley, ``Products and convolutions of {Gaussian} probability density
  functions,'' \emph{Tina-Vision Memo}, vol.~3, no.~4, p.~1, 2003.

\bibitem{bishop2006pattern}
C.~M. Bishop, \emph{Pattern recognition and machine learning}.\hskip 1em plus
  0.5em minus 0.4em\relax Springer, 2006.

\bibitem{liu2015optimal}
W.~Liu, ``Optimal filtering for discrete-time linear systems with
  time-correlated multiplicative measurement noises,'' \emph{IEEE Transactions
  on Automatic Control}, vol.~61, no.~7, pp. 1972--1978, 2015.

\bibitem{yang2013optimal}
Z.~Yang, X.~Shi, and J.~Chen, ``Optimal coordination of mobile sensors for
  target tracking under additive and multiplicative noises,'' \emph{IEEE
  Transactions on Industrial Electronics}, vol.~61, no.~7, pp. 3459--3468,
  2013.

\bibitem{yu2021robust}
X.~Yu and Z.~Meng, ``Robust {Kalman} filters with unknown covariance of
  multiplicative noise,'' \emph{arXiv:2110.08740}, 2021.

\bibitem{2015Convergence}
B.~Chen, J.~Wang, H.~Zhao, N.~Zheng, and J.~C. Principe, ``Convergence of a
  fixed-point algorithm under maximum correntropy criterion,'' \emph{IEEE
  Signal Processing Letters}, vol.~22, no.~10, pp. 1723--1727, 2015.

\bibitem{yu2021adaptive}
X.~Yu and J.~Li, ``{Adaptive Kalman filtering for recursive both additive noise
  and multiplicative noise},'' \emph{IEEE Transactions on Aerospace and
  Electronic Systems}, vol.~58, no.~3, pp. 1634--1649, 2022.

\end{thebibliography}

\end{document}